\documentclass[a4paper,onecolumn,10pt,accepted=2024-03-24]{quantumarticle}
\pdfoutput=1
\usepackage[utf8]{inputenc}
\usepackage[english]{babel}
\usepackage[T1]{fontenc}
\usepackage{amsmath}

\usepackage{tikz}
\usepackage{lipsum}

\usepackage{graphicx}
\usepackage{tabularx}
\usepackage{braket}
\usepackage{amssymb}
\usepackage{bm}
\usepackage[driverfallback=dvipdfm]{hyperref}
\hypersetup{colorlinks=true,breaklinks,urlcolor=blue,linkcolor=blue,citecolor=blue}

\newcommand{\red}[1]{#1}

\newcommand{\comment}[1]{}

\newcommand{\lr}[1]{\left( #1\right)}
\newcommand{\mlr}[1]{\left[ #1\right]}
\newcommand{\glr}[1]{\left\{ #1\right\}}
\newcommand{\alr}[1]{\left\langle #1\right\rangle}
\newcommand{\norm}[1]{\left\lVert#1\right\rVert}
\newcommand{\abs}[1]{\left\lvert#1\right\rvert}
\newcommand{\rarrow}{\quad \Rightarrow \quad}
\newcommand{\ii}{\mathrm{i}}
\newcommand{\ee}{\mathrm{e}}
\newcommand{\dd}{\mathrm{d}}
\newcommand{\tr}[1]{\mathrm{tr}\lr{#1}}

\newcommand{\diam}{\mathrm{diam}}
\newcommand{\where}{\quad {\rm where}\quad}
\newcommand{\order}{\mathrm{O}}

\newcommand{\GHZ}{\mathrm{GHZ}}
\newcommand{\fX}{\mathbf{X}}
\newcommand{\id}{\mathrm{id}}
\newcommand{\cZ}{\mathcal{Z}}
\newcommand{\floor}[1]{\left\lfloor#1\right\rfloor}

\newcommand{\cM}{\mathcal{M}}
\newcommand{\tcM}{\widetilde{\cal M}}
\newcommand{\cL}{\mathcal{L}}
\newcommand{\brao}[1]{\left( #1 \right |}
\newcommand{\keto}[1]{\left| #1 \right )}

\newcommand{\V}{{\mathbf{V}}}
\newcommand{\poly}[1]{\mathrm{poly}\lr{#1} }
\newcommand{\fkd}{\mathfrak{d}}

\renewcommand{\thetable}{\arabic{table}}
\renewcommand{\thesection}{\arabic{section}}
\renewcommand{\thesubsection}{\thesection.\arabic{subsection}}
\renewcommand{\thesubsubsection}{\thesubsection.\arabic{subsubsection}}
\makeatletter
\renewcommand{\p@subsection}{}
\renewcommand{\p@subsubsection}{}
\makeatother

\usepackage{amsthm}
\newtheorem{thm}{Theorem}
\numberwithin{thm}{section}
\newtheorem{cor}[thm]{Corollary}
\newtheorem{lem}[thm]{Lemma}

\newtheorem{prop}[thm]{Proposition}

\newtheorem{proc}[thm]{Procedure}

\begin{document}

\title{Heisenberg-limited metrology with perturbing interactions}

\author{Chao Yin}\email{chao.yin@colorado.edu}
\affiliation{Department of Physics and Center for Theory of Quantum Matter, University of Colorado, Boulder CO 80309, USA}
\orcid{0000-0003-3379-310X}

\author{Andrew Lucas}\email{andrew.j.lucas@colorado.edu}
\affiliation{Department of Physics and Center for Theory of Quantum Matter, University of Colorado, Boulder CO 80309, USA}

\maketitle

\begin{abstract}
    We show that it is possible to perform Heisenberg-limited metrology on GHZ-like states, in the presence of generic spatially local, possibly strong interactions during the measurement process. An explicit protocol, which relies on single-qubit measurements and feedback based on polynomial-time classical computation, achieves the Heisenberg limit.  In one dimension, matrix product state methods can be used to perform this classical calculation, while in higher dimensions the cluster expansion underlies the efficient calculations.   The latter approach is based on an efficient classical sampling algorithm for short-time quantum dynamics, which may be of independent interest.
\end{abstract}

\maketitle

\tableofcontents

\section{Introduction}
\red{In classical physics, an ideal detector and noise-free experiment can perfectly measure the state of a system.   In quantum physics, this is fundamentally not possible: due to Heisenberg's uncertainty principle, we cannot deterministically measure the state of a spin-1/2 system without knowing the axis along which its spin $\mathbf{S}$ is aligned.  The rich field of quantum metrology has been developed over the previous decades in order to understand both the fundamental limits on how many measurements are needed to measure some parameter to a given accuracy in quantum mechanics, as well as to develop protocols that can achieve such fundamental limits: see \cite{metro_rev14,metro_rev11,sensing_rmp} for reviews.}

\red{As we will review below in more detail, it is well-established that using unentangled states of $N$ spins, after measuring $M$ times some unknown parameter $\omega$, the uncertainty $\delta \omega \gtrsim (MN)^{-1/2}$.  This is called the standard quantum limit, and it arises from the \emph{classical} central limit theorem on the measurement outcomes of unentangled qubits.   However, using a cleverly chosen entangled quantum state, one can improve this sensitivity: $\delta \omega \gtrsim M^{-1/2}N^{-1}$.  Such scaling is called the Heisenberg limit.  Thus, quantum entanglement can enhance the sensitivity of our measuring apparatus.}

\red{Unfortunately, an entangled state of $N$ qubits is \emph{extraordinarily fragile} to perturbations, so it is unlikely that such a state can plausibly ever be built in experiment.  It is therefore of critical importance to understand \emph{to what extent the Heisenberg limit is robust to perturbations}.  This paper will address from one particular perspective, in which an experimentalist is handed a highly entangled state capable of achieving the Heisenberg limit, yet in which the measurement procedure itself is imperfect: the qubits interact with themselves in addition to sensing the external parameter $\omega$.  Similar problems were studied in \cite{onsiteV13,onsiteV14,metro_Ising15,metro_Choi17,metro_PT18,metro_strong18,metro_domino21,metro_all2all22,metro_scar21,metro_frag22,metro_singlep_power22}.  We will show that without proper ``error correction" for these interactions, the experimentalist will lose any advantage to using the entangled quantum state.  At the same time, we prove that there is a \emph{classically computable} protocol involving measurement and feedback, which enables the experimentalist to achieve Heisenberg-limited scaling. Our result proves a conjecture from this earlier literature (see e.g. \cite{onsiteV14}) that Heisenberg-limited metrology is robust in the presence of \emph{known} unitary perturbations.}

\red{The remainder of the introductory section overviews our results in a more thorough fashion, with explicit formulas outlining our setup and main results, together with some background into the field of metrology.  The remainder of the paper rigorously demonstrates our claims.}

\subsection{Heisenberg-limited metrology in an ideal world}\label{sec:ideal}
Consider sensing a parameter $\omega$ (often a magnetic field along a specified axis) using $N$ qubits (i.e. spins-$\frac{1}{2}$), whose Pauli matrices are denoted by $X_j, Y_j, Z_j$ where $j\in \Lambda := \{1,\cdots,N\}$. The spins start in some initial state $\rho_{\rm in}$, evolve under the Hamiltonian \begin{equation}
    H_\omega^{\text{ideal}} = \omega Z:=\omega \sum_{j=1}^N Z_j,
\end{equation} for some time $t$, and are then measured.  The goal is to use the measurement outcomes to learn the parameter $\omega$.

One repeats this procedure $M$ times: preparing the same initial state $\rho_{\rm in}$, and running the same protocol each time.  After this, one uses the collective information about all the measurement outcomes to deduce an optimal  estimate $\omega_{\mathrm{est}}$.

If the initial state is unentangled, the best precision $\lr{\delta\omega}^2:= \mathbb{E}[(\omega_{\mathrm{est}}-\omega)^2]$ one can achieve is given by the \emph{standard quantum limit} (SQL) \begin{equation}\label{eq:SQL}
    \delta\omega_{\rm SQL}
    =  \lr{MNt^2}^{-1/2}.
\end{equation}
Here $\mathbb{E}[\cdots]$ denotes an average over all the possible measurement outcomes in each of the $M$ trials. The $N^{-1/2}$ scaling arises
because the $N$ spins undergo independent dynamics, and the measurement outcomes are uncorrelated. Here and below, we refer to reviews \cite{metro_rev14,metro_rev11,sensing_rmp} for further introduction to the subject of metrology.  While this result also scales as $M^{-1/2}$ because each trial is independent and classical post-processing of the noisy measurement outcomes is constrained by the central limit theorem of probability, this $M$ scaling is unavoidable so long as one must repeat the experiment $M$ times using the same $N$ qubits.   The main question then becomes: can we improve the $N$ scaling of (\ref{eq:SQL})?

The answer is yes. If we begin with an entangled initial state --  the Greenberger–Horne–Zeilinger (GHZ) state $\rho_{\rm in} = \ket{\GHZ}\bra{\GHZ}$, where  \begin{equation}\label{eq:GHZ}
    \ket{\GHZ} := \frac{1}{\sqrt{2}} \lr{\ket{0\cdots 0} + \ket{1\cdots 1}},
\end{equation}
it is possible to parametrically reduce $\delta\omega$ at large $N$.  To see how, observe that if we evolve for time $t$ using $H_\omega^{\mathrm{ideal}}$, this state evolves to \begin{equation}\label{eq:psi_id}
    \ket{\psi_{\id}} := \ee^{-\ii t\omega Z} \ket{\GHZ} = \frac{1}{\sqrt{2}} \lr{\ee^{\ii N\omega t}\ket{0\cdots 0} + \ee^{-\ii N\omega t}\ket{1\cdots 1}},
\end{equation}
gaining an \emph{extensive phase} difference between the two parts in \eqref{eq:GHZ}.
Then by measuring the expectation value of observable $\fX:=\prod_j X_j$: \begin{equation}\label{eq:Xid}
    \alr{\fX}_{\id} := \bra{\psi_{\id}} \fX \ket{\psi_{\id}} = \cos(2N\omega t),
\end{equation}
one achieves the \emph{Heisenberg limit} (HL) \begin{equation}\label{eq:HL0}
    \delta\omega_{\rm HL} = \lr{4MN^2t^2}^{-1/2},
\end{equation}
which scales as $\sim N^{-1}\ll N^{-1/2}$, which is a dramatic advantage over SQL.

To derive \eqref{eq:HL0}, observe that for large $M$, \begin{equation}\label{eq:HL}
    (\delta\omega)^2 \approx \frac{1}{M} \frac{\lr{\Delta \fX}_{\id}^2}{\abs{\partial_\omega \alr{\fX}_{\id}}^2} = \frac{1}{M} \frac{\sin^2(2N\omega t)}{\abs{2Nt \sin(2N\omega t)}^2} = \frac{1}{4MN^2t^2}.
\end{equation}
See Fig.~\ref{fig:slope}(a) for an illustration of the first equality.
Here we have used \eqref{eq:Xid}, and the fact that the variation for Pauli-like (involutory) observable $\fX$  obeys \begin{equation}\label{eq:DeltaX}
    \lr{\Delta \fX}^2_\id := \alr{\fX^2}_{\id} -\alr{\fX}_{\id}^2 = 1- \alr{\fX}_{\id}^2.
\end{equation}

Beyond its theoretical desirability, this entanglement-enhanced metrology is efficient to implement (in the absence of quantum noise or errors!): the global operator $\fX$ can be measured simply by making simultaneous projective single-qubit measurements in the local $X$-basis.   The needed $\langle \fX\rangle$ is just the  parity $\pm 1$ of the product of the $N$ measurement outcomes, averaged over $M$ trials.

Notice that in the discussion above, we cannot fully determine the rotation angle $\theta=2N\omega t$ using the algorithm as stated.  The reason is that we cannot tell apart $\theta$ and $\theta+2\pi$. Happily, the problem is easy to correct.  We will assume, throughout this paper, that the ``integer part'' \begin{equation}\label{eq:w'=}
    \omega' = \frac{\pi}{2Nt} \lr{\floor{\frac{2N\omega t}{\pi}}+\frac{1}{2}}.
\end{equation} 
of $\omega$ is known, such that \begin{equation}\label{eq:w-w'}
    \omega \in \mathcal{I}_{\omega'}:= \mlr{\omega'- \frac{\pi}{4Nt}, \omega'+ \frac{\pi}{4Nt}}.
\end{equation} 
We learn $\omega^\prime$ by, for example, evolving first for an extremely short time $t$, such that we can deduce the first digit in $\omega$ (assuming we have an order of magnitude estimate for its value!).   Then, we double the time, such that we can get a slightly more accurate estimate of $\omega$; iterating this process many times, we can obtain as accurate of an estimate as we want, until we hit the HL (at which point the best thing to do is take $M\rightarrow \infty$). See \cite{bitwise_sense09,bitwise_sense15,bitwise_sense20} for precise statements on this procedure; in particular, the precision is of Heisenberg scaling in terms of the overall resources. 

\subsection{The problem of interest}
\label{sec:problem}
\begin{figure}[t]
\centering
\includegraphics[width=.85\textwidth]{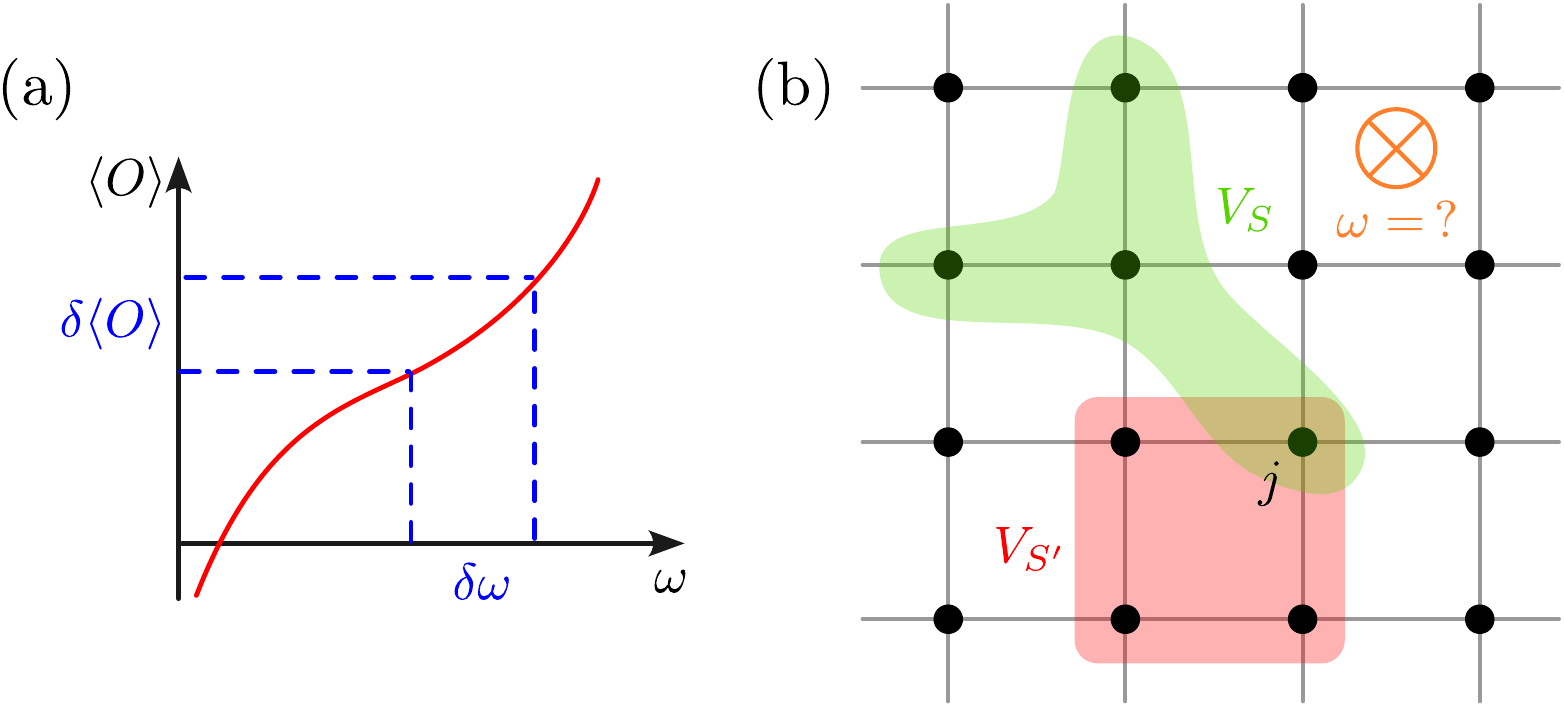}
\caption{\label{fig:slope} (a) Given measurement of observable $O$ with precision $\delta\alr{O}=\Delta O/\sqrt{M}$, it leads to a precision $\delta\omega$ determined by the slope $\partial_\omega \alr{O}$. See \eqref{eq:HL} and \eqref{eq:DeltaX} for the case of $O=\bf{X}$ in the ideal GHZ metrology. (b) A sketch of qubit interactions. In this example, the interaction graph $G$ is the 2d square lattice, so the distance between two qubits $\mathsf{d}(\cdot,\cdot)$ is the Manhattan distance. Two spatially local terms in $V$, $V_S$ and $V_{S'}$, are shown in the two shaded regions, each coupling four qubits. They both contribute to the local strength defined in \eqref{eq:J=} at qubit $j$. One can check that they have diameters $\mathrm{diam}(S)=3$ and $\mathrm{diam}(S')=2$. }
\end{figure}

In this paper, we study the robustness of the above GHZ metrology to known perturbations.  In particular, we assume the Hamiltonian acting on the spins is \begin{equation}\label{eq:Hw}
    H_\omega = V+\omega Z,
\end{equation}
which contains unwanted local interaction $V$ among the spins of local strength $J$. To be precise, the $N$ spins sit on the vertices of a graph $G$,
connected by edges that define a distance function $\mathsf{d}(\cdot,\cdot)$ on the graph. We assume each vertex connects to at most $K$ other vertices by edges. Then we assume the interaction is \begin{equation}
    V = \sum_{S\subseteq \Lambda: \diam(S)\le \ell} V_S,
\end{equation}
where \red{$V_S$ acts nontrivially on set $S$ of qubits,} $\diam(S) = \max_{i,j\in S} \mathsf{d}(i,j)$ is the diameter of set $S$, and $\ell$ is the interaction range. The local strength is defined by \begin{equation}\label{eq:J=}
    J:=\max_{j\in \Lambda} \sum_{S\ni j} \norm{V_S},
\end{equation}
\red{which quantifies how strong the interaction on a single qubit can be. Here $\norm{\cdot}$ denotes the operator norm, i.e. the largest singular value of the operator. As an example of the above formal definitions, $V$ can be local interactions on a constant-dimensional lattice as sketched in Fig.~\ref{fig:slope}(b), where each $V_S=\sum_P a_P P$ can expanded by the Pauli strings contained in the local set $S$, with coefficients $a_P=\order(J)$}.\footnote{\red{Here and throughout the paper, $f=\order(g)$ means there exists a constant $c$ determined by $K$ and $\ell$ such that $f\le cg$. We use $\Omega$ to denote the similar statement with $\ge$ instead, and $f=\Theta(g)$ if both $f=\order(g)$ and $f=\Omega(g)$ hold.}}
Note that $J$ could depend on how the total operator $V$ is decomposed into local terms $V_S$; since our results apply to any such decomposition, the reader can just choose a convenient decomposition, although in practice the results will be stronger for a decomposition that achieves the smallest possible $J$. We in general do not require $J$ to be small compared to $\omega$, although requiring so indeed gives a longer time scale for achieving HL, as we will see.   

\red{An explicit example of a Hamiltonian that is within the purview of the theory we describe below would be a transverse-field Ising model in a strong magnetic field: \begin{equation}
    H_\omega = V+\omega Z = \sum_{i=1}^{N-1} J X_i X_{i+1} + \omega \sum_{i=1}^N Z_i.
\end{equation}
 In an experiment, such a situation in \eqref{eq:Hw} may arise when there are spin-spin interactions that one cannot perfectly cancel (e.g. by dynamical decoupling \cite{dyn_decoup99}).\footnote{In many atomic systems, these spin-spin interactions could be long-ranged; this is a technical complication that we will not address in this work.}} If one wants to measure the magnetic field in a very local spatial region, one needs to spatially move the $N$ spins closer together, and $V$ may no longer be negligible.  If $V$ is neglected altogether, it could in principle become the dominant noise in an experiment.

Besides changing the Hamiltonian, we also address more general initial states $\ket{\psi}_{\rm in}$ beyond the ideal GHZ state. We evolve the system under \eqref{eq:Hw} by a time $t$ independent of $N$, and ask whether measuring the final state \begin{equation}\label{eq:psiw}
    \ket{\psi_\omega} := \ee^{-\ii t H_\omega}\ket{\psi}_{\rm in},
\end{equation}
can achieve HL, in terms of its $N$-scaling $\delta\omega \propto N^{-1}$, given the exact knowledge of $V$. 

\red{We will assume $V$ is known exactly.  The HL is impossible to achieve if each local $V_S$ has uncertainty $\Omega( 1/N)$. For example, suppose one pretends that there is no interaction, but there is actually a weak interaction: \begin{equation}
    V= \frac{1}{N}\sum_{j=1}^{N-2}Z_j Z_{j+1} Z_{j+2}
\end{equation}. When evolving the GHZ state, $V$ contributes a $\Theta(1)$ phase difference between the two parts in \eqref{eq:psi_id}, which leads to $\Theta(1/N)$ biased error for $\omega$ that shares the same precision of the prior knowledge \eqref{eq:w-w'} and cannot be reduced by increasing $M$.  Therefore, an unknown $V$ can be just as dangerous as $Z$-type errors for metrology (the rate of which cannot exceed $1/N$ \cite{channel_estimate21}). Nevertheless, our results could be useful so long as $V$ is a well-calibrated interaction.}

\subsection{Overview of our results}
The first question one might ask is whether or not the same protocol sketched in (\ref{eq:Xid}) would apply.  Unfortunately, the answer is no: the naive protocol (\ref{eq:Xid}) is extremely sensitive to perturbation $V$, even if $V$ is an on-site field: see Appendix \ref{app:A}. Some intuition for this result is that an \emph{unknown} $V$ is no better than decohering noise, which immediately leads to the SQL \cite{metro_noise11,metro_noise12}. Evidently, more work will be required to achieve the Heisenberg limit in the presence of $V$.  In this subsection, we both outline the rest of the paper, and explain in less technical terms our main results along with the intuition for why they hold. We focus on the ideal GHZ initial state for this section.


\subsubsection{Robustness of the ``phase difference" to perturbations}\label{sec:phase_intro}
In the ideal case, the extensive phase difference \eqref{eq:psi_id} comes from the extensive difference in $Z$ polarization of the two parts: $\bra{0\cdots0} Z\ket{0\cdots 0} - \bra{1\cdots1} Z\ket{1\cdots 1}=2N$. See the green curves in Fig.~\ref{fig:result}(a), where $\ket{\alpha}_{\rm in}$ represents $\ket{\alpha\cdots \alpha}$ (or its generalized version, see Section \ref{sec:QFI}). With interactions, the two parts are evolved at time $t$ to $\ket{\phi^0_\omega}$ and $\ket{\phi^1_\omega}$ that are no longer the maximally polarized product states. Nevertheless, the extensive $Z$ difference persists for short time: \begin{equation}\label{eq:Z-Z_intro}
    \bra{\phi^0_\omega} Z\ket{\phi^0_\omega}-\bra{\phi^1_\omega}Z\ket{\phi^1_\omega}=\Theta(N), \quad \mathrm{if}\quad t< c_{\rm in}/J,
\end{equation}
\red{for some O(1) constant $c_{\rm in}$,}
as shown by the distance between the two red curves in Fig.~\ref{fig:result}(a).
The reason is simple: for each qubit, it needs a nonvanishing time $\Omega(1/J)$ to change its local state (polarization) significantly, because it only couples to $\order(1)$ other qubits nearby \red{on the interaction graph. More precisely, this can be derived from \begin{equation}
    \norm{\frac{\dd}{\dd t} \rho_i} = \norm{\mathrm{tr}_{i^{\rm c}}\lr{ [H_\omega, \rho]}} =\norm{\mathrm{tr}_{i^{\rm c}}\lr{ [H_{\text{on } i}, \rho]}} \le 2\lVert H_{\text{on } i}\rVert \cdot \lVert \rho\rVert = \order(J),
\end{equation}
where $\rho_i=\mathrm{tr}_{i^{\rm c}}(\rho)$ ($\rm c$ means complement) is the local density matrix of qubit $i$, and $H_{\text{on } i}$ contains the finite number of terms in $H_\omega$ that act nontrivially on $i$; other terms do not contribute due to the cyclic property of the partial trace.}

\eqref{eq:Z-Z_intro} implies that the two parts keep gaining an extensive phase difference for short time. More precisely, when tuning $\omega$, the phase difference changes accordingly by an extensive \red{$\order(N)$} sensitivity.
In this process, since $\ket{\phi^\alpha_\omega}$ ($\alpha=0,1$) also rotate in the physical Hilbert space (which is absent in the ideal case), it is desirable that such rotation is \emph{subdominant} comparing to the change of phase. We show this from the fact that $\ket{\phi^\alpha_\omega}$ is short-time evolution from a product state and thus short-range correlated. The intuition is that a local Hamiltonian like $H_\omega$ maps a product state (more generally, a short-range correlated state) to an orthogonal one with amplitude bounded by $\order(\sqrt{N}) \ll N\sim \norm{H_\omega}$.
In Fig.~\ref{fig:result}(a), this is manifested by the $\order(\sqrt{N})$ width of the wave-packets:\begin{equation}\label{eq:cor_intro}
    \bra{\phi^\alpha_\omega} Z^2\ket{\phi^\alpha_\omega} - \bra{\phi^\alpha_\omega} Z\ket{\phi^\alpha_\omega}\bra{\phi^\alpha_\omega} Z\ket{\phi^\alpha_\omega} = \order(N),
\end{equation}
for any $\alpha=0,1$.

Gathering the two ingredients \eqref{eq:Z-Z_intro} and \eqref{eq:cor_intro}, we prove in Section \ref{sec:QFI} (which contains precise statements) that HL is robust for \red{$t< c_{\rm in}/J$}, because the extensive phase difference \eqref{eq:psi_id} is well-defined in the presence of interaction and can be measured in principle to estimate $\omega$.

\subsubsection{Measurement and feedback to achieve the Heisenberg limit}

Now that we know there is a coherent phase difference between the two halves of the initial GHZ state even in the presence of interactions, we must now develop a protocol to measure the state and this relative phase difference accumulated.  The key challenge is that if we do not measure in the right basis, we will accidentally measure which of the two $|\phi^{0,1}_\omega\rangle$ we are in (analogous to a $Z$ error destroying all entanglement in GHZ).  Hence, we must now develop efficient protocols to measure the phase difference over all $N$ qubits, without causing such a generalized $Z$ error, in order to achieve HL for \red{$t<c_{\rm in} /J$}.
Such protocols are desirable: the $x$-basis measurement in the ideal case no longer works as shown in Appendix \ref{app:A},\footnote{For $V=JX$ considered there, local measurements along a tilted axis (instead of along $x$) restore HL. However, for general $V$ that is not on-site, we expect HL cannot be achieved by any local-basis measurement, where the local basis can vary qubit-by-qubit; we were, however, unable to prove this statement rigorously. In other words, we expect classical communication is necessary in the LOCC measurement protocol developed later in this paper.} and one does not want to implement a protocol that uses an \emph{exponential} (in system size $N$) amount of classical/quantum resources. 

If one can engineer generic Hamiltonian evolution in the quantum experiment, we show in Appendix \ref{app:B} that one can effectively\footnote{Note that one cannot perfectly reverse the time evolution because $\omega$ is an unknown parameter, but we show that this is not required to achieve HL.} reverse the evolution by $H_\omega$, after which an $x$-basis measurement like the ideal case leads to HL. However, such quantum control is typically demanding in experiment. 

Therefore, we focus on another protocol based on local operations and classical communication (LOCC). As depicted in Fig.~\ref{fig:result}(b), the qubits are measured one-by-one, where the measurement basis depends on previous measurement outcomes (i.e., classical communication) and is determined by classical computation. During the measurement protocol, we assume the qubits undergo no Hamiltonian evolution or decoherence. In practice, this condition holds if measurement and classical computation is much faster than Hamiltonian evolution etc, or if one deliberately separates the qubits after the $H_\omega$ evolution to well-isolated ones, awaiting for measurement. 

\begin{figure}[htbp]
\centering
\includegraphics[width=.85\textwidth]{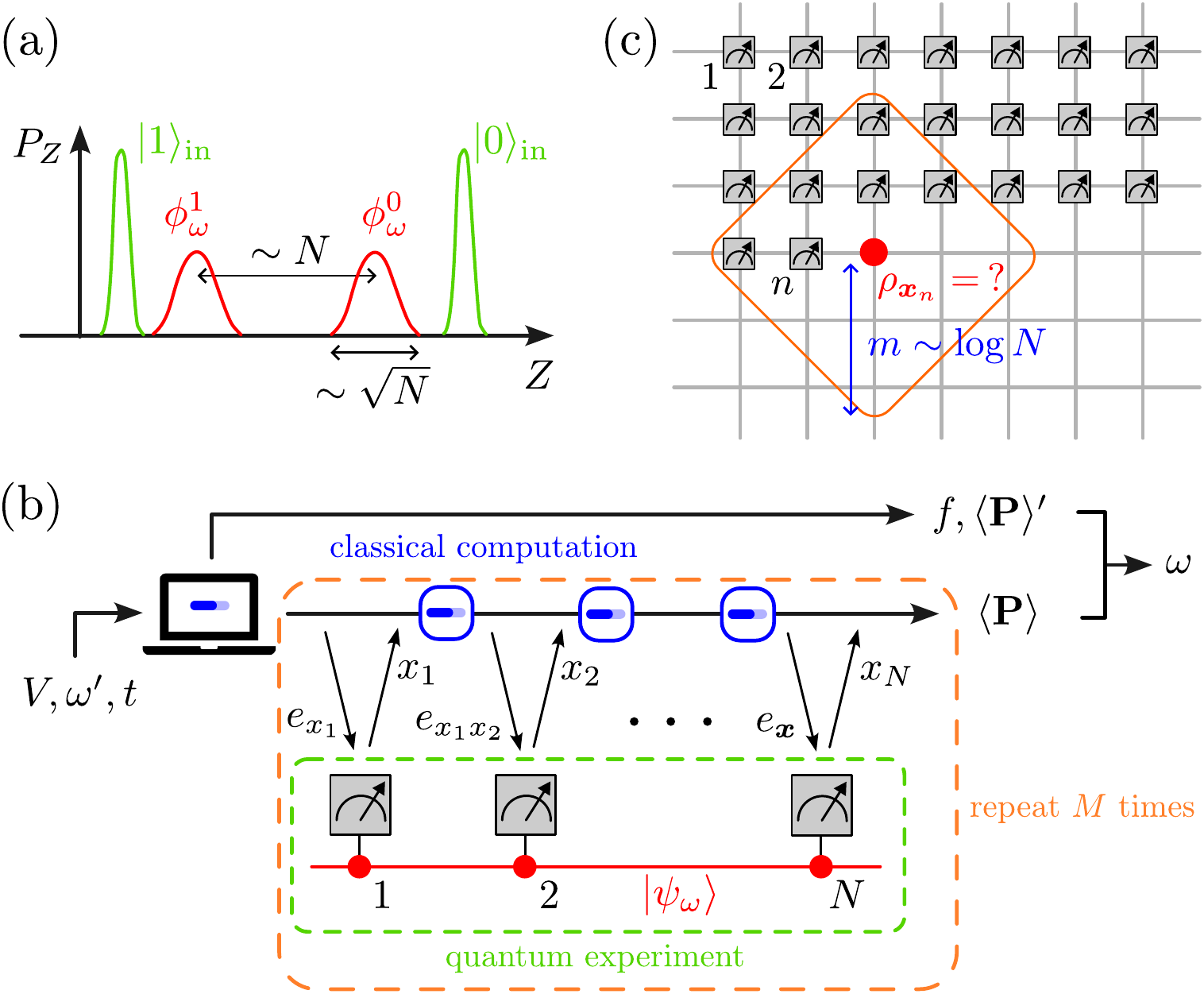}
\caption{\label{fig:result} Sketch of our main results. (a) An illustration for robustness of HL, using distribution $P_Z$ in $Z$ polarization. We assume the initial state (green curves) is like $\ket{\GHZ}$, where the two parts $\ket{0}_{\rm in}$ and $\ket{1}_{\rm in}$ have an extensive $\sim N$ difference in $Z$, and each of them has small $\sim\sqrt{N}$ fluctuation. We prove these properties hold after evolving for time \red{$t< c_{\rm in}/J$} (red curves), so that the two parts robustly gain an extensive phase difference leading to HL. If $J\ll \omega$, this robustness extends to any $N$-independent time $t$. (b) A sketch of the LOCC measurement protocol that achieves HL for \red{$t<c_{\rm M}/J$}. Although classical computation is needed to determine the local measurement basis like $\ket{e_{x_1x_2}}$, we prove the overhead is $\poly{N}$ for many situations. In the final step to extract the paramter $\omega$, the average parity $\alr{\mathbf{P}}$ of the measurement string $\bm{x}$ is compared to two quantities $f,\alr{\mathbf{P}}'$ from classical computation, using \eqref{eq:P-P=f}. (c) A sketch of the $\poly{N}$-time classical sampling algorithm at short time \red{$t<c_{+}/J$} (under certain conditions), which also serves as an ingrediant for proving efficiency of the LOCC measurement protocol in (b). To compute a local density matrix $\rho_{\bm{x}_n}$ after measuring $n$ qubits (so that one can determine the measurement probability of qubit $n+1$), we prove that it only depends on dynamics within distance $m\sim \log N$ of the given vertex. In the figure $m=2$. Constraining computation to the radius-$m$ ball (shown in orange) naively leads to a quasi-polynomial $\exp[(\log N)^d]$ complexity for $d$-dimensional systems; we use the cluster expansion to get $\poly{N}$ for any $d$, and even for any bounded degree graph. Note that measurements in the figure occur one-by-one in a spatially continuous way: This is only for simplicity of drawing and not required in general. }
\end{figure}

Such an LOCC measurement protocol for pure-state metrology was previously proposed in \cite{Zhou_LOCC20}, where they specified an algorithm to find the measurement basis (using previous measurement outcomes) that achieves the best possible precision. We adapt this algorithm to our problem starting from Section \ref{sec:1d}, which achieves HL due to results described in Section \ref{sec:phase_intro}. However, \cite{Zhou_LOCC20} did not study the classical complexity of the algorithm, and in general one might expect it to be exponentially difficult to find this ideal basis due to the many-body Hilbert space involved.  However, we will prove that the classical computation of this efficient basis can be done using \emph{polynomial classical resources and runtime}.   This is usually denoted as $\poly{N}$ (i.e. bounded by $c_q N^q$ for some constants $q$ and $c_q$).     

\subsubsection{Efficient classical sampling of short time dynamics}\label{sec:samp_intro}
If the qubits are arranged in one dimension, the computation is done efficiently using matrix product state (MPS) representations for the two parts $\ket{\phi^\alpha_\omega}$ of the state, because they are both short-time evolution from a product state. This is detailed in Section \ref{sec:1d}. For higher dimensions or general graph $G$, MPS techniques do not work, and we need to develop another method.
The classical computation in the LOCC algorithm turns out to be merely the same as another task of its own interest, namely classically sampling LOCC measurements on $\ket{\phi^0_\omega}$. 

More precisely, consider evolving an initial product state using a local Hamiltonian like \eqref{eq:Hw} for time $t$, after which the system is measured in the computational basis  to output a classical string $\bm{x}$ with probability $p_{\bm{x}}$. It is crucial to understand whether this quantum experiment can be classically simulated in polynomial time, i.e., whether a classical computer can also sample the string $\bm{x}$ efficiently. It is believed (based on conjectures in theoretical computer science) that a depth-3 quantum circuit in 2d, which can be viewed as a time-dependent Hamiltonian evolution with $t=\Theta(1)$, already becomes exponentially hard to classically sample in the worst case \cite{depth3_hard04}. However, it is unclear whether this quantum advantage holds for any constant $t$ independent of $N$. Moreover, generalizing computational-basis measurement, it is sometimes desirable to perform LOCC measurement to try to do measurement-based quantum computation (MBQC) on the final state \cite{MBQC_rev09}. Is this adaptive case even harder to classically simulate? Is it capable to do universal quantum computation like MBQC on the cluster state \cite{oneway_qc01}?

To answer these questions, in Section \ref{sec:samp} we prove that for \red{$t< c_{+}/J$ with a constant $c_{+}$}, the final state can be sampled by a $\poly{N}$ classical algorithm.  Therefore it is a computationally simple state, and does not enable universal MBQC. Technically, we achieved this by generalizing the cluster expansion method developed in \cite{learn_highT21,Loschmidt_echo23}, and rely on one assumption that the measurement basis is not so close to that of the initial state. The idea is shown in Fig.~\ref{fig:result}(c): the local state on an unmeasured qubit is easy to simulate, because it only depends on dynamics (and measurement) in a neighboring region of size $\sim \log N$. A rough physical picture for this is as follows. After evolution in time $t= 0.1/J$, each pair of neighboring qubits only establishes $\sim 0.1$ correlation. Suppose qubits $1,2$ and $2,3$ are such pairs, but $1,3$ does not share correlation. Although measuring $2$ may help to create correlation between $1$ and $3$, the amount is $0.1^2=0.01$. Repeating this argument implies that faraway qubits have correlation exponentially small in distance, even after some local measurements are done.

Based on the complexity result in Section \ref{sec:samp}, we return to the metrology problem in Section \ref{sec:>1d}, and prove (without assumptions on the measurement basis) that the LOCC protocol achieving HL involves $\poly{N}$ classical computation, for general graphs and \red{$t<c_{\rm M}/J$ for some constant $c_{\rm M}$}.

\subsubsection{Robustness against weak perturbations}

The above results hold even if the interaction is strong $J\gg \omega$. If it is actually weak $J\ll \omega$, we show in Section \ref{sec:preth} simply by energy conservation that HL is robust well beyond the time window \red{$Jt< c_{\rm in}$} in Section \ref{sec:phase_intro}. However, for such long times, we no longer have guarantees on the efficiency of the measurement protocol in general. We also discuss
prethermalization theory \cite{preth_rigor17} that plays a role in preserving $Z$ polarization against small perturbations.

\subsection{Further background on metrology: quantum Fisher information}
The rest of the introduction provides further background knowledge, along with context, for our work. In this subsection, we summarize further useful facts about quantum metrology which are well-established in the literature. The first question we ask is whether the state \eqref{eq:psiw}
could achieve HL, without restriction on the kinds of measurements made (or classical algorithm used to process them). This question is equivalent to asking for the scaling of the quantum Fisher information (QFI) of $\ket{\psi_\omega}$ with respect to $\omega$: For a pure state, QFI is defined as \cite{QFI76,QFI94,Nk_limit07} \red{(see e.g. Eq.(4) in \cite{metro_noise13})} \begin{equation}\label{eq:QFI}
    {\cal F}(\ket{\psi_\omega}) = 4 \norm{\lr{1-\ket{\psi_\omega}\bra{\psi_\omega}}\partial_\omega\ket{\psi_\omega}}^2 = 4\lr{ \braket{ \partial_\omega \psi_\omega | \partial_\omega \psi_\omega} - \abs{\braket{ \psi_\omega | \partial_\omega \psi_\omega}}^2},
\end{equation}
The first expression in \eqref{eq:QFI} is intuitive: QFI simply measures how fast $\ket{\psi_\omega}$ rotates in the physical Hilbert space when tuning $\omega$, where the unphysical global phase does not contribute. On the other hand, the second expression in \eqref{eq:QFI} can be understood as some connected correlation function (see \eqref{eq:F<N} below). QFI bounds the estimation precision by the quantum Cram\' er-Rao bound (QCRB) \begin{equation}\label{eq:QCRB}
    (\delta\omega)^2 \ge \frac{1}{M \cal F(\ket{\psi_\omega})}.
\end{equation}
In the $M\gg 1$ regime that we focus here, \eqref{eq:QCRB} is saturable \cite{QFI94,estimation_tech09} by projective measurements in the eigenbasis of \begin{equation}\label{eq:Lw}
    L_\omega = 2\lr{1-\ket{\psi_\omega}\bra{\psi_\omega}}\ket{\partial_\omega \psi_\omega}\bra{\psi_\omega} + \rm H.c.,
\end{equation}
followed by classical post-processing the data using maximum likelihood estimation. See \cite{pi_correct20} for modification of QCRB in the case $M=1$.

In general, to achieve HL it is necessary to have ${\cal F}(\ket{\psi_\omega}) \propto N^2$, which we prove in Corollary \ref{cor:QFI} and Proposition \ref{prop:cor_decay} for the problem described in Section \ref{sec:problem}.  However, bounding QFI is not the only way to show robustness of HL. Indeed the most direct result we will establish is (as summarized above), based on proving that the extensive phase difference between the two halves of the GHZ state is robust to arbitrary perturbations (at short $t$).  While this fact does imply the desired scaling of QFI, we will argue that the constraints of locality at short time scales in quantum mechanics ensure more than simply good QFI: they also ensure an efficient measurement protocol to achieve the HL in metrology.   In particular, it is quite undesirable to do an eigenbasis measurement of \eqref{eq:Lw}, because: (\emph{i}) it requires evolving the many-body state $|\psi_\omega\rangle$ which a priori could be exponentially hard in $N$, (\emph{ii}) the naive algorithm of just measuring $L_\omega$ requires knowledge of $\omega$ anyway, which is precisely what we want to learn, and (\emph{iii}) measuring $L_\omega$ is not realistic in any experiment with $N\gtrsim 4$.  In fact, property (\emph{ii}) means the measurement protocol does only \emph{local} quantum estimation \cite{estimation_tech09}. 

Ultimately, the QCRB can be saturated (or at least saturated up to a constant factor) by more than one protocol.  The goal of our work is to find \emph{efficient} ones with \emph{global} quantum estimation, like the ideal one in Section \ref{sec:ideal}, which involves only local measurements and applies to all $\omega$ in range $\mathcal{I}_{\omega'}$ \eqref{eq:w-w'}.

\subsection{Previous work on robustness of metrology}
Special cases of our problem described in Section \ref{sec:problem} have been considered in the literature. 
\cite{onsiteV13,onsiteV14} show that HL is robust if $V$ is on-site. In particular, \cite{onsiteV13} \emph{conjectures} the robustness holds for more general interactions, which is  proven in this paper.
\cite{metro_prl11,metro_Ising15,metro_Choi17,metro_PT18,metro_strong18,metro_domino21,metro_all2all22,metro_scar21,metro_frag22,metro_singlep_power22} consider specific models, where the chosen interaction $V$ sometimes enhances the precision or robustness of metrology. For general unwanted interactions, it is proposed to reduce the interaction strength by dynamical decoupling \cite{metro_dd13,metro_dd16,metro_Lukin20}, where the qubits are actively operated by control pulses.
Indeed, quantum control is shown to have advantages in the setting of estimating \emph{multiple} parameters \cite{metro_multip_rev16}, especially learning a many-body Hamiltonian \cite{learn_control23,learn_err23}. In this respect, our result may be surprising: HL for single-parameter estimation is actually robust even without quantum control during sensing. The price to pay is the nontrivial (but provably efficient) LOCC measurement procedure \emph{after} sensing, in order to accurately estimate the parameter.

We assume the system is isolated from the environment throughout the paper. If there is decoherence from coupling to the environment, HL is not robust anymore and reduces to SQL in general \cite{metro_noise11,metro_noise12}, even if the qubits couple to the environment independently. Intuitively, this comes from the fragility of the global many-body entanglement of the GHZ state. In certain cases, HL can be restored by active quantum error correction (QEC) \cite{metro_QEC14_1,metro_QEC14_2,metro_QEC14_3,metro_QEC17,metro_QEC18,metro_QEC23}.
However, it is usually assumed that the QEC operation is much faster than the decoherence rate. Our results may shed light on the actual QEC time scale needed, especially for decoherence that are correlated among qubits \cite{corre_noise14,corre_noise18,corre_noise19,corre_noise20,corre_noise22}.

\section{Robustness of the Heisenberg limit}\label{sec:QFI}
Having summarized our strategy, we can now begin by deriving our first key result: the Heisenberg limit is robust \red{if $Jt$ is smaller than a constant} even in the presence of strong perturbations.  Note that this result does not rule out the possibility of HL robustness on longer time scales (see e.g. Section \ref{sec:preth}).  Still, as we will see, this more limited result will be adequate for our purposes.

\subsection{Generalized initial states}
We allow for any initial state of the form \begin{equation}\label{eq:psi_in}
    \ket{\psi}_{\rm in}:=\lr{\ket{0}_{\rm in} + \ket{1}_{\rm in} }/\sqrt{2},
\end{equation}
where the two parts $\ket{\alpha}_{\rm in}$ ($\alpha=0,1$) are orthonormal and satisfy the following two conditions: \begin{enumerate}
    \item They have extensive $Z$ polarization difference: \begin{equation}\label{eq:Z0-Z1=cin}
    \alr{Z}_0-\alr{Z}_1 = 2c_{\rm in} N,\where \alr{O}_\alpha := \braket{\alpha|O|\alpha}_{\rm in}.
\end{equation}
Here $\alpha=0,1$, $O$ is any operator, and $c_{\rm in}>0$ is a constant.

\item They are short-range correlated: \begin{equation}\label{eq:initial_cor}
    \abs{\alr{O_A O_B}_\alpha - \alr{O_A}_\alpha \alr{O_B}_\alpha } \le c_\xi \ee^{-\mathsf{d}(A,B)/\xi},
\end{equation}
for any subsets $A,B\subset \Lambda$, and operators in them with $\norm{O_A}=\norm{O_B}=1$. Here $\mathsf{d}(\cdot, \cdot)$ is the distance function defined in Section \ref{sec:problem}, and $\xi$ is the correlation length. If the graph $G$ is not finite-dimensional, we assume the correlation length is relatively short \begin{equation}\label{eq:xi<K}
    \xi < 1/\log(K-1).
\end{equation}

\end{enumerate}  
The GHZ state is thus the extreme case $c_{\rm in}=1$ and $\xi=0$. We expect these conditions are satisfied by many more states of physical interest: e.g. equal superposition of $\mathbb{Z}_2$ symmetry broken states, or ``rotated'' version of GHZ state like $\ket{0\cdots 0} + \lr{\tilde{\alpha}\ket{0} + \tilde{\beta}\ket{1}}^{\otimes N}$ (with $\abs{\tilde{\alpha}}^2 + |\tilde{\beta}|^2=1$, $|\tilde{\beta}|>0$) where the two parts have negligible $\sim \exp(-N)$ overlap and can be massaged to form \eqref{eq:psi_in}.

The time-evolved state $\ket{\psi_\omega}$ \eqref{eq:psiw} is then \begin{equation}\label{eq:psi=phi}
    \ket{\psi_\omega} = \frac{1}{\sqrt{2}} \lr{\ket{\phi^0_\omega}+\ket{\phi^1_\omega} },\where \ket{\phi^\alpha_\omega} := \ee^{-\ii t H_\omega}\ket{\alpha}_{\rm in}.
\end{equation}

\subsection{Robustness of the extensive phase difference}
We first show that before some $\order(1)$ time, the two parts in \eqref{eq:psi=phi} keep gaining an extensive phase difference when varying $\omega$, just like the ideal case.

\begin{thm}\label{thm:phase}
If the initial state satisfies \eqref{eq:Z0-Z1=cin} and \eqref{eq:initial_cor}, then \begin{equation}\label{eq:partial_phi}
    \partial_\omega \ket{\phi^\alpha_\omega} = -\ii\mlr{c'+(-1)^\alpha c_\omega }N\ket{\phi^\alpha_\omega} + \order (\sqrt{N}),
\end{equation}
where \begin{equation}\label{eq:cw>}
    c_\omega\ge 
    t(c_{\rm in}-Jt).
\end{equation}
As a result, the whole state satisfies (dropping the unphysical global phase) \begin{equation}\label{eq:psi'_omega}
    \ket{\psi_\omega} =\frac{1}{\sqrt{2}} \lr{\ee^{-\ii f(\omega)/2} \ket{\phi^0_{\omega'}} + \ee^{\ii f(\omega)/2} \ket{\phi^1_{\omega'}} } + \order(1/\sqrt{N}),
\end{equation}
where the function $f(\omega): \mathcal{I}_{\omega'} \rightarrow \mlr{-\frac{\pi}{2}, \frac{\pi}{2}}$ satisfies $f(\omega')=0$, and is an increasing function for $t<c_{\rm in}/J$ with slope proportional to $N$: \begin{equation}\label{eq:dfw_N}
    2t(c_{\rm in}-Jt)N\le \partial_\omega f(\omega)\le 2tN.
\end{equation}
\end{thm}

We will provide the proof shortly.
Without QFI, \eqref{eq:psi'_omega} is already transparent on why HL is achievable for constant time $t<c_{\rm in}/J$, even for global quantum estimation: the two states $\ket{\phi^0_\omega}$ and $\ket{\phi^1_\omega}$ just gain opposite phases that are proportional to $N$ when tuning $\omega$. Suppose the function $f$ is known, then for $t<c_{\rm in}/J$ up to a negligible error $\order(N^{-1/2})$, a protocol achieves HL as long as it measures the relative phase\footnote{For example, one can measure the observable $\ket{\phi^0_{\omega'}}\bra{\phi^1_{\omega'}}+\ket{\phi^1_{\omega'}}\bra{\phi^0_{\omega'}}$ which does not depend on $\omega$ but only on the prior knowledge $\omega'$. However, this is a nonlocal operator in general that is hard to implement experimentally, so starting from next section we will find efficient measurement protocols. } between the two states with outcome $f_{\rm E}$ (a number), because one can then solve $f(\omega_{\rm est})=f_{\rm E}$ to get an estimate $\omega_{\rm est}$ of the true $\omega$. As we will show, the function $f$ will be efficiently computable by a classical computer: one does not need to compute the value $f(\tilde{\omega})$ for every $\tilde{\omega}$; it suffices if for any given $\tilde{\omega}\in \mathcal{I}_{\omega'}$, the classical algorithm outputs $f(\tilde{\omega})$ with good accuracy. The reason is that one can easily perform a classical binary search algorithm when comparing to the quantum experimental value $f_{\rm E}$:

\begin{prop}\label{prop:binary}
    Suppose there is an oracle such that for each given $\tilde{\omega}\in \mathcal{I}_{\omega'}$, it outputs an approximation $f_{\rm O}(\tilde{\omega})$ with error \begin{equation}\label{eq:dfw<M}
        \abs{f_{\rm O}(\tilde{\omega})- f(\tilde{\omega})} \le 1/\sqrt{M}.
    \end{equation} Then based on the quantum experimental result $f_{\rm E}$, calling the oracle $\order(\log M)$ times gives an estimate $\omega_{\rm est}$ for the true $\omega$ with HL.
\end{prop}
\begin{proof}
    After $M$ measurements, the precision $\abs{f_{\rm E} - f(\omega)}\le \pi/(2\sqrt{M})$, because the range of function $f$ is bounded by $|f(\omega)|\le \pi/2$. As a result, since the oracle has similar precision \eqref{eq:dfw<M}, it suffices to find an $\omega_{\rm est}$ such that the oracle output $f_{\rm O}(\omega_{\rm est})$ is within distance \red{$\order(1/\sqrt{M})$} to the experimental value $f_{\rm E}$. This would achieve HL due to \eqref{eq:dfw_N}.

    To find $\omega_{\rm est}$, we use the following binary search algorithm. If one is lucky that \red{$\abs{f_{\rm E}} = \order( 1/\sqrt{M})$}, then $\omega_{\rm est}=\omega'$. Otherwise depending on whether $f_{\rm E}$ is positive or negative (we assume positive for simplicity), calculate $f_{\rm O}(\omega'+\pi/(4Nt))$ and compare it to $f_{\rm E}$: if it is within distance \red{$\order(1/\sqrt{M})$} then $\omega_{\rm est}=\omega'+\pi/(4Nt)$; otherwise call the oracle at the middle point $\omega'+\pi/(8Nt)$ of the candidate interval, and repeat the above process. The length of the candidate interval shrinks exponentially with the number of steps, so that after $\order(\log M)$ calls of the oracle, one is guaranteed to get an $\omega_{\rm est}$ with HL \eqref{eq:HL}.
\end{proof}

For completeness, we show that \eqref{eq:partial_phi} indeed implies HL for QFI.

\begin{cor}\label{cor:QFI}
If the initial state satisfies \eqref{eq:Z0-Z1=cin} and \eqref{eq:initial_cor}, then QFI defined in \eqref{eq:QFI} satisfies
\begin{equation}\label{eq:F>N2}
    {\cal F}(\ket{\psi_\omega}) \ge 4t^2N^2(c_{\rm in}-Jt)^2 + \order(N^{3/2}),
\end{equation}
which scales as $\sim N^2$ so long as $t<c_{\rm in}/J$.   Note that we assume this short time inequality when stating (\ref{eq:F>N2}). 
\end{cor}

\begin{proof}
We derive from the bound \eqref{eq:partial_phi}. From \eqref{eq:QFI} and \eqref{eq:psi=phi}, \begin{align}\label{eq:F>partial_phi}
    {\cal F}(\ket{\psi_\omega}) &= 2\sum_{\alpha,\beta=0,1} \braket{ \partial_\omega \phi^\beta_\omega | \partial_\omega \phi^\alpha_\omega} + \lr{\sum_{\alpha,\beta=0,1}\braket{ \phi^\beta_\omega | \partial_\omega \phi^\alpha_\omega}}^2 \nonumber\\ &= 2c_\omega^2 N^2 \sum_{\alpha,\beta=0,1}(-1)^{\alpha-\beta} \braket{ \phi^\beta_\omega | \phi^\alpha_\omega} - c_\omega^2N^2\lr{\sum_{\alpha,\beta=0,1}(-1)^{\alpha} \braket{ \phi^\beta_\omega | \phi^\alpha_\omega}}^2 + \order(N^{3/2}) \nonumber\\
    &\ge 4c_\omega^2N^2+ \order(N^{3/2}) \ge 4t^2N^2(c_{\rm in}-Jt)^2+ \order(N^{3/2}).
\end{align}
Here in the second line, we have used $\norm{\partial_\omega\ket{\phi^\alpha_\omega}} =\order(N)$ according to \eqref{eq:partial_phi}, so that the second term in \eqref{eq:partial_phi} only contributes to $\order(N^{3/2})$ here. In the last line we have used $\braket{\phi^\beta_\omega|\phi^\alpha_\omega}=\braket{\beta|\alpha}_{\rm in}=\delta_{\beta\alpha}$ and \eqref{eq:cw>}.
\end{proof}

\subsection{Proof of Theorem \ref{thm:phase}}

We first sketch the idea, focusing on \eqref{eq:partial_phi} since \eqref{eq:psi'_omega} will follow by a simple integration over $\omega$. The $c'$ term in \eqref{eq:partial_phi} is an unimportant global phase. The $c_\omega$ term, on the other hand, is the crucial \emph{relative} phase between the two parts $\alpha=0,1$. Intuitively, an extensive $Z$ polarization leads to an extensive relative phase gained by tuning $\omega$, so we will first show the initial extensive $\alr{Z}$ difference \eqref{eq:Z0-Z1=cin} persists in short time. This is summarized in Lemma \ref{prop:extensive_phase}, which does not rely on the condition on correlation \eqref{eq:initial_cor}. Due to the presence of $V$, the two parts do not simply gain phases, but also rotate in the physical Hilbert space (where states differing by an overall phase are identified). Nevertheless, since the two parts are initially short-range correlated \eqref{eq:initial_cor}, the state rotated by a single term in the local Hamiltonian is orthogonal to each other. As a simple example, $X_i\ket{0\cdots 0} \perp X_j\ket{0\cdots 0}$ for $i\neq j$. Therefore, the norm \emph{squared} of the rotated part of the state scales linearly with $N$. This leads to the last term in \eqref{eq:partial_phi}, which is subdominant $\sim \sqrt{N}$ comparing to the relative phase $\sim N$. We note that similar locality estimates are also used to prove that even if the Hamiltonian is interacting, SQL cannot be surpassed by separable initial states \cite{metro_sep_bound23}.

We now show that the extensive phase difference $c_\omega$ term in \eqref{eq:partial_phi} comes from condition \eqref{eq:Z0-Z1=cin} alone. 

\begin{lem}\label{prop:extensive_phase}
The two components in \eqref{eq:psi=phi} satisfy 
    \begin{equation}\label{eq:0-1>N}
    \ii \lr{ \bra{\phi^0_\omega}\partial_\omega \ket{\phi^0_\omega}-\bra{\phi^1_\omega}\partial_\omega \ket{\phi^1_\omega}} \ge 2t(c_{\rm in}-Jt)N.
\end{equation}
This does not rely on condition \eqref{eq:initial_cor}.
\end{lem}

\begin{proof}

Taking derivative on the matrix $\ee^{-\ii t H_\omega}$ with respect to $\omega$, we have \begin{equation}\label{eq:cZ=Z}
    \partial_\omega \ket{\phi^\alpha_\omega} =-\ii \int^t_0 \dd s\, \ee^{-\ii (t-s) H_\omega} Z \ee^{-\ii s H_\omega} \ket{\alpha}_{\rm in} = -\ii t\ee^{-\ii t H_\omega}\overline{Z} \ket{\alpha}_{\rm in},\where \overline{Z} = \frac{1}{t} \int^t_0 \dd s\, \ee^{\ii s H_\omega} Z \ee^{-\ii s H_\omega}.
\end{equation}
Therefore \begin{equation}\label{eq:partial=Z}
    \ii\bra{\phi^\alpha_\omega}\partial_\omega \ket{\phi^\alpha_\omega} = t \bra{\alpha}_{\rm in} \ee^{\ii t H_\omega} \ee^{-\ii t H_\omega}\overline{Z} \ket{\alpha}_{\rm in} = t\braket{\overline{Z}}_\alpha.
\end{equation} 
Using \eqref{eq:Z0-Z1=cin}, \begin{align}\label{eq:>N-Z-Z}
    \ii\lr{ \bra{\phi^0_\omega}\partial_\omega \ket{\phi^0_\omega}-\bra{\phi^1_\omega}\partial_\omega \ket{\phi^1_\omega}} &= t\lr{\braket{\overline{Z}}_0 - \braket{\overline{Z}}_1} \nonumber\\ &= t\lr{\braket{Z}_0 - \braket{Z}_1+ \braket{\overline{Z}-Z}_0 - \braket{\overline{Z}-Z}_1}\nonumber\\ &\ge 2t (c_{\rm in}N-\norm{\overline{Z}-Z}).
\end{align}
In general, the time-averaged operator $\overline{Z}$ is close to $Z$ in short time: \begin{align}\label{eq:Z-Z<}
    \norm{\overline{Z}-Z} &\le \frac{1}{t} \int^t_0 \dd s\, \norm{\ee^{\ii s H_\omega} Z \ee^{-\ii s H_\omega}-Z} \nonumber\\ &= \frac{1}{t} \int^t_0 \dd s\, \norm{\int^s_0 \dd s'\ee^{\ii s' H_\omega} [H_\omega, Z]\ee^{-\ii s' H_\omega} } \nonumber\\
    &\le \frac{1}{t} \int^t_0 \dd s\, s\norm{[H_\omega, Z]}= \frac{t}{2}\norm{[V, Z]} = \frac{t}{2} \norm{\sum_j \mlr{\sum_{S\ni j} V_S,Z_j} } \le NJt,
\end{align}
where we have used \eqref{eq:J=} and $\norm{[A,B]}\le2 \norm{A}\norm{B}$. This establishes \eqref{eq:0-1>N} due to \eqref{eq:>N-Z-Z}.
\end{proof}

\begin{proof}[Proof of Theorem \ref{thm:phase}]
Since we have proven the extensive relative phase \eqref{eq:0-1>N}, to prove \eqref{eq:partial_phi} it remains to show the extra rotation in the physical Hilbert space contributes $\order(\sqrt{N})$ in \eqref{eq:partial_phi}:
\begin{equation}\label{eq:partial_phi<N}
    \bra{\partial_\omega \phi^\alpha_\omega}\lr{1-\ket{\phi^\alpha_\omega}\bra{\phi^\alpha_\omega}} \ket{\partial_\omega \phi^\alpha_\omega} \le Nt^2 g(Jt), 
\end{equation}
where the function $g(Jt)$ is independent of $N$, and $t^2$ comes from dimensional analysis due to two $\omega$-derivatives on the left hand side. We have projected out the part in $\partial_\omega \ket{\phi^\alpha_\omega}$ proportional to $\ket{\phi^\alpha_\omega}$, which would contribute to the $c_\omega$ term in \eqref{eq:partial_phi}. One can verify that combining \eqref{eq:0-1>N} and \eqref{eq:partial_phi<N} yields \eqref{eq:partial_phi}.
\red{The reason is as follows. One can always decompose \begin{equation}\label{eq:partial=decompose}
    \ket{\partial_\omega \phi^\alpha_\omega} = \ket{\phi^\alpha_\omega}\braket{\phi^\alpha_\omega|\partial_\omega \phi^\alpha_\omega} + \lr{1-\ket{\phi^\alpha_\omega}\bra{\phi^\alpha_\omega}} \ket{\partial_\omega \phi^\alpha_\omega},
\end{equation}
where \eqref{eq:partial_phi<N} implies that $\lVert \lr{1-\ket{\phi^\alpha_\omega}\bra{\phi^\alpha_\omega}} \ket{\partial_\omega \phi^\alpha_\omega} \rVert = \order(t\sqrt{N})$ and thus corresponds to the second term of $\order(\sqrt{N})$ in \eqref{eq:partial_phi}. For the first term of \eqref{eq:partial_phi}, the $c'$ term just corresponds to a global phase, and is unimportant. Taking the inner product between \eqref{eq:partial=decompose} and $\bra{\phi^\alpha_\omega}$, the $c_\omega$ term is further extracted by subtracting the two choices of $\alpha$, which leads to \eqref{eq:cw>} from \eqref{eq:0-1>N}. }

To show \eqref{eq:partial_phi<N}, we first rewrite its left hand side using \eqref{eq:cZ=Z}, \eqref{eq:partial=Z} and $\braket{\partial_\omega \phi^\alpha_\omega| \partial_\omega \phi^\alpha_\omega}=t^2 \bra{\alpha}_{\rm in} \overline{Z} \ee^{\ii t H_\omega} \ee^{-\ii t H_\omega}\overline{Z} \ket{\alpha}_{\rm in}=t^2\braket{\overline{Z}^2}_\alpha$. The result is\footnote{Note that (\ref{eq:F<N}) is proportional to the QFI of $\ket{\phi^\alpha_\omega}$ as a function of $\omega$.} 
\begin{align}\label{eq:F<N}
    t^2\lr{ \braket{\overline{Z}^2}_\alpha-\braket{\overline{Z}}^2_\alpha} = t^2\sum_{i} \sum_j \lr{ \braket{\overline{Z}_i\overline{Z}_j}_\alpha-\braket{\overline{Z}_i}_\alpha \braket{\overline{Z}_j}_\alpha }\le t^2 N \max_i \sum_j \lr{ \braket{\overline{Z}_i\overline{Z}_j}_\alpha-\braket{\overline{Z}_i}_\alpha \braket{\overline{Z}_j}_\alpha }.
\end{align}
In some sense, the final sum quantifies the total correlation shared by spin $i$ with all other spins, in the final state $\ket{\phi^\alpha_\omega}$. Intuitively, since in finite time $i$ can only build correlation effectively with a finite set of neighboring spins, we expect \begin{equation}\label{eq:cor<g}
    \sum_j\braket{\overline{Z}_i\overline{Z}_j}_\alpha-\braket{\overline{Z}_i}_\alpha \braket{\overline{Z}_j}_\alpha\le g(Jt),
\end{equation}
i.e., bounded by some function independent of $N$, which leads to \eqref{eq:partial_phi<N}. 

More rigorously, since $V$ is local, we have Lieb-Robinson bound \cite{Lieb1972}: \red{There exists an operator $Z_{i,r}(t)$ supported inside the ball $\mathcal{B}_{i,r}:=\{k\in \Lambda: \mathsf{d}(k,i)\le r\}$ centered at $i$ of radius $r$, such that} \begin{equation}\label{eq:LRB}
    \norm{Z_i(t)-Z_{i,r}(t)}\le c_{\rm LR} \ee^{\mu(vt-r)},\quad \forall r\ge 1,
\end{equation}
where $\norm{Z_{i,r}(t)}\le 1$, and $c_{\rm LR},\mu, v>0$ are constants.\footnote{\red{$Z_{i,r}(t)$ can be obtained by projecting $Z_i(t)$ onto its components supported inside the ball $\mathcal{B}_{i,r}$.; see Propositions 3.5 and 4.1 in \cite{review_us23}.}} The bound \eqref{eq:LRB} holds similarly for \red{the time-averaged operator $\overline{Z}_i(t)$, because one can bound the approximation error of the integral in \eqref{eq:cZ=Z} by a triangle inequality like the first line of \eqref{eq:Z-Z<}}. Then using the initial condition \eqref{eq:initial_cor}, connected correlation for any pair $i,j$ is bounded by an exponential-decaying function \begin{equation}\label{eq:ZiZj<}
    \braket{\overline{Z}_i\overline{Z}_j}_\alpha-\braket{\overline{Z}_i}_\alpha \braket{\overline{Z}_j}_\alpha \le (c_\xi + 6 c_{\rm LR}) \exp\lr{-\frac{\mathsf{d}(i,j)-2vt }{\xi+2\mu^{-1}} },
\end{equation}
following
\cite{LRB_cor06} (see also Theorem 5.9 in \cite{review_us23}). 

Now sum \eqref{eq:ZiZj<} over $j$. The result is independent of $N$ for finite-dimensional lattices, because the exponential decay in distance $\mathsf{d}(i,j)$ overcomes the polynomial number of $j$s at the given distance. For general graphs with bounded degree $K$, the number of spins at a given distance to $i$ is upper bounded by $K(K-1)^{\mathsf{d}(i,j)-1}$, so the decay in \eqref{eq:ZiZj<} still wins due to \eqref{eq:xi<K}, and the fact that $\mu$ can be chosen arbitrarily large in \eqref{eq:LRB} (by adjusting $c_{\rm LR}$ and $v$ accordingly) \cite{review_us23}. This establishes \eqref{eq:cor<g} and further \eqref{eq:partial_phi<N} from \eqref{eq:F<N}. 

We have proven \eqref{eq:partial_phi} from \eqref{eq:0-1>N} and \eqref{eq:partial_phi<N}. It remains to integrate it from $\omega'$ to $\omega$ to derive \eqref{eq:psi'_omega}. Since \red{$\abs{\omega-\omega'}=\order( N^{-1})$} \eqref{eq:w-w'}, the $\order(\sqrt{N})$ term in \eqref{eq:partial_phi} translates to the $\order(1/\sqrt{N})$ term in \eqref{eq:psi'_omega}. At the same time, $f(\omega)=\int^\omega_{\omega'} c_{\omega''} \dd\omega''$, so $f(\omega')=0$, and the first inequality of \eqref{eq:dfw_N} simply comes from \eqref{eq:partial_phi}. The second inequality of \eqref{eq:dfw_N} comes from the first line of \eqref{eq:>N-Z-Z}, which implies $c_\omega \le tN$ because \begin{equation}\label{eq:cZ<N}
    \norm{\overline{Z}}\le t^{-1} \int^t_0 \dd s \norm{Z}=\norm{Z}=N.
\end{equation}
Due to this bound on the slope, the range of $f(\omega)$ is contained in $[-\pi/2, \pi/2]$.
\end{proof}

\subsection{An alternative viewpoint: correlation cannot decay in $\order(1)$ time }

Theorem \ref{thm:phase} above is one central result of this paper: it shows the robustness of the relative phase for GHZ-like states, proves HL for QFI in \eqref{eq:F>N2}, and sets foundations for the latter sections on efficient protocols. As an interlude, this subsection is devoted to an \emph{alternative} argument for robustness of metrology, at least at the level of QFI. The intuition is that the super-sensitivity of $\ket{\GHZ}$ with respect to the $Z$ field stems from its large $ZZ$ correlation/fluctuation. Mathematically, one can verify that QFI \eqref{eq:QFI} is just proportional to $ZZ$ connected correlation function (see \eqref{eq:ZZ=cin}) when $V=0$. Then for robustness, we basically need to show such long-range correlation cannot be killed in $\order(1)$ time evolution. Formalizing this idea, we have:

\begin{prop}\label{prop:cor_decay}
    If initially \begin{equation}\label{eq:ZZ=cin}
        \alr{ZZ}_{\rm in}-\alr{Z}_{\rm in}\alr{Z}_{\rm in}=c'_{\rm in} N^2,
    \end{equation} 
    with $c'_{\rm in}>0$, then QFI defined in \eqref{eq:QFI} satisfies \begin{equation}\label{eq:QFI>c'}
        {\cal F}(\ket{\psi_\omega}) \ge 4t^2 N^2 (c'_{\rm in}-4Jt ),
    \end{equation}
    which is HL for $t<c'_{\rm in}/(4J)$.
\end{prop}

The GHZ state corresponds to the case $c'_{\rm in}=1$. We do not need the conditions \eqref{eq:Z0-Z1=cin} and \eqref{eq:initial_cor} here, although it might be possible to derive \eqref{eq:ZZ=cin} from those (possibly under physical assumptions that off-diagonal elements like $\alr{0|Z|1}_{\rm in}$ are subdominant). Beyond GHZ metrology, Proposition \ref{prop:cor_decay} thus also applies to metrology using spin-squeezed states \cite{squeez_rev11}. However, we leave as future work to design efficient measurement protocols for such states.

\begin{proof}
Using \eqref{eq:cZ=Z} in a similar way as \eqref{eq:F<N}, QFI in \eqref{eq:QFI} becomes
    \begin{align}\label{eq:F>N2-N}
    \frac{1}{4t^2}{\cal F}(\ket{\psi_\omega}) &= \braket{\overline{Z}^2}_{\rm in}-\braket{\overline{Z}}^2_{\rm in} \nonumber\\ &= \braket{Z^2}_{\rm in}-\braket{Z}^2_{\rm in} + \braket{Z(\overline{Z}-Z)+(\overline{Z}-Z)Z+(\overline{Z}-Z)^2}_{\rm in} - 2\braket{Z}_{\rm in}\braket{\overline{Z}-Z}_{\rm in}-\braket{\overline{Z}-Z}^2_{\rm in} \nonumber\\
    &\ge c'_{\rm in}N^2 - 4N \norm{\overline{Z}-Z}.
\end{align}
Here in the second line we have used $\overline{Z}=Z+(\overline{Z}-Z)$. In the third line, we have used \eqref{eq:ZZ=cin}, $\braket{(\overline{Z}-Z)^2}_{\rm in}\ge \braket{\overline{Z}-Z}^2_{\rm in}$, $\abs{\braket{Z}_{\rm in} }\le N$ and $\braket{Z(\overline{Z}-Z)}_{\rm in}\ge -\norm{Z}\norm{\overline{Z}-Z}=-N\norm{\overline{Z}-Z}$. Plugging in \eqref{eq:Z-Z<} leads to \eqref{eq:QFI>c'}.
\end{proof}

\section{LOCC protocol in 1d: sampling matrix product states}\label{sec:1d}
\subsection{LOCC protocol}

For our situation where the system remains in a pure state, \cite{Zhou_LOCC20} shows that QCRB can always be saturated by a measurement protocol built of local operation and classical communication (LOCC), where the spins $\{1,2,\cdots,N\}$ are measured one after another, adaptively. Namely, first measure spin $1$ in some chosen basis, and then measure spin $2$ in a basis depending on the previous measurement outcome $x_1\in \{0,1\}$, and so on. After measuring all the spins, one obtains a binary string $\bm{x}:=x_1\cdots x_N$, from which one tries to decode the parameter $\omega$. Note that spins with nearby labels do not need to be close spatially. The measurement is described by the set of projectors \red{$E_{\bm{x}}$, or equivalently the set of product states $\ket{E_{\bm{x}}}$:} \begin{subequations}\begin{align}\label{eq:Ex_POVM}
    E_{\bm x} &= e_{x_1} \otimes e_{x_1x_2}\otimes \cdots \otimes e_{x_1\cdots x_N }, \\
    \ket{E_{\bm x}} &= \ket{e_{x_1}} \otimes \ket{e_{x_1x_2}}\otimes \cdots \otimes \ket{e_{x_1\cdots x_N } }
\end{align}\end{subequations}
where $e_{x_1\cdots x_j}=\ket{e_{x_1\cdots x_j}} \bra{e_{x_1\cdots x_j}}$ is a projector on spin $j$ that projects onto a basis \red{$\{\ket{e_{x_1\cdots x_{j-1}0}}, \ket{e_{x_1\cdots x_{j-1}1}}\}$} determined by the previous measurement outcomes $x_1\cdots x_{j-1}$. 
Like in measurement-based quantum computation, such classical communication is necessary in general \cite{Zhou_LOCC20}; see also \cite{Friedman:2022vqb}. Note that we use lowercase $e$ for operator/state on a single qubit, while uppercase $E$ for projector on multiple qubits. 

For $t< c_{\rm in}/J$, combining Corollary \ref{cor:QFI} with \cite{Zhou_LOCC20} implies the existence of a LOCC protocol that achieves HL, at least for local estimation.
However, although \cite{Zhou_LOCC20} describes a procedure to obtain the local basis $E_{\bm{x}}$ defining the protocol based on the quantum state, it is not guaranteed that this procedure is easy to implement in practice. More precisely, we need to show that the local basis is \emph{classically computable} in $\poly{N}$ runtime (which implies space complexity is also $\poly{N}$). From now on, we show that there is indeed a LOCC protocol of global estimation that (\emph{i}) achieves HL for certain class of GHZ-like initial states with $t< c_{\rm in}/J$, and (\emph{ii}) is provably efficient to implement. In this section, after showing the procedure to determine the local basis \eqref{eq:Ex_POVM}, we focus on 1d where the complexity result is relatively easy to establish, and applies to more general initial states, using matrix product states (MPSs). In the next two sections, we use sampling complexity results to deal with general interaction graphs.

\subsection{Finding the LOCC measurement basis}

In this subsection, we adapt the procedure in \cite{Zhou_LOCC20} to our case of GHZ metrology with global estimation. Namely, we show how to determine the local measurement basis $E_{\bm{x}}$ from knowledge of $\ket{\psi_{\omega'}}$ where $\omega'$ is known.
Observe that the probability distribution $\bra{\psi_{\omega'}} E_{\bm x} \ket{\psi_{\omega'}}$ of the measurement string for $\ket{\psi_{\omega'}}$ is \begin{equation}\label{eq:P'x=}
    \bra{\psi_{\omega'}} E_{\bm x} \ket{\psi_{\omega'}} = \frac{1}{2}\lr{\braket{E_{\bm x}}_{\phi^0_{\omega'}} + \braket{E_{\bm x}}_{\phi^1_{\omega'}} } + \mathrm{Re} \bra{\phi^0_{\omega'}} E_{\bm x} \ket{\phi^1_{\omega'}},
\end{equation}
Since the relative phase between $\phi^0_{\omega'}$ and $\phi^1_{\omega'}$ only contributes to the second term, the intuition is to make the second term large enough, and check that is also the case for any $\omega\in \mathcal{I}_{\omega'}$. Indeed, if the measurement basis is not chosen carefully, the second term would become exponentially small, leading to SQL as shown in Proposition \ref{prop:X_fail}. Strictly speaking, the second term has amplitude $\abs{\bra{\phi^0_{\omega'}} E_{\bm x} \ket{\phi^1_{\omega'}}}$ and phase $\theta'_{\bm x}=\mathrm{arg} \bra{\phi^0_{\omega'}} E_{\bm x} \ket{\phi^1_{\omega'}}$, and we demand a criterion for each: \begin{enumerate}
    \item We want the amplitude to be as large as possible: In order to saturate Cauchy-Schwarz inequality \red{\begin{align}\label{eq:01<0+1}
        \abs{\bra{\phi^0_{\omega'}} E_{\bm x} \ket{\phi^1_{\omega'}}} = \abs{\bra{\phi^0_{\omega'}} E_{\bm x}^{1/2}\cdot  E_{\bm x}^{1/2} \ket{\phi^1_{\omega'}}} &\le \frac{1}{2}\lr{\bra{\phi^0_{\omega'}} E_{\bm x}^{1/2}\cdot  E_{\bm x}^{1/2} \ket{\phi^0_{\omega'}} + \bra{\phi^1_{\omega'}} E_{\bm x}^{1/2}\cdot  E_{\bm x}^{1/2} \ket{\phi^1_{\omega'}}}\nonumber\\ &= \frac{1}{2}\lr{\braket{E_{\bm x}}_{\phi^0_{\omega'}} + \braket{E_{\bm x}}_{\phi^1_{\omega'}} },
    \end{align}}
    we desire $E^{1/2}_{\bm x} \ket{\phi^0_{\omega'}}$ and $E^{1/2}_{\bm x} \ket{\phi^1_{\omega'}}$ only differ by a phase, or equivalently, \begin{equation}\label{eq:EME=0'}
        \bra{E_{\bm x}} \cM \ket{ E_{\bm x}}=0,\where \cM := \ket{\phi^0_{\omega'}}\bra{\phi^0_{\omega'}}-\ket{\phi^1_{\omega'}}\bra{\phi^1_{\omega'}}.
    \end{equation}

    \item The phase \begin{equation}\label{eq:sign=-x}
        \ee^{\ii\theta'_{\bm x}}=(-1)^{\bm x}\ii=\pm \ii,
    \end{equation}
    where the $+$ ($-$) sign is for even (odd) string $\bm x$, i.e., determined by the parity. This requires $\bra{\phi^0_{\omega'}} E_{\bm x} \ket{\phi^1_{\omega'}}$ is pure imaginary, or equivalently, \begin{equation}\label{eq:EtME=0}
        \bra{E_{\bm x}} (\tcM+\tcM^\dagger) \ket{E_{\bm x}}=0, \where \tcM := \ket{\phi^1_{\omega'}}\bra{\phi^0_{\omega'}}.
    \end{equation}
\end{enumerate} 

These two conditions combine to \begin{equation}\label{eq:1=-x0}
    \braket{E_{\bm x}| \phi^1_{\omega'}} = (-1)^{\bm x}\ii \braket{E_{\bm x}| \phi^0_{\omega'}}.
\end{equation}
Note that the phase \eqref{eq:sign=-x} is chosen deliberately such that \eqref{eq:1=-x0} is satisfied by the standard $x$-basis measurement in the unperturbed case $V=0$, where a known field $\omega'$ \eqref{eq:w'=} would accumulate a phase $\theta =\pi/2$ modular $\pi$.
In general, the measurement basis satisfying \eqref{eq:1=-x0} can be determined by the following procedure that generalizes \cite{Zhou_LOCC20}.
\begin{proc}\label{proc:Ex} 
Procedure to determine the local basis $e_{x_1},e_{x_1x_2},\cdots$ along a given trajectory $x_1,x_2,\cdots$.
\begin{enumerate}
    \item Calculate ``reduced matrices'' on spin $1$: $\cM_{\varnothing}=\mathrm{tr}_{2\cdots N}{\cM}, \tcM_{\varnothing}=\mathrm{tr}_{2\cdots N}{\tcM}$. Here subscript $\varnothing$ means no previous substring.
    \item Find an orthonormal basis $\ket{e_{x_1}}$ for spin $1$ s.t. \begin{equation}\label{eq:E1ME1'=0}
        \bra{e_{x_1}} \cM_{\varnothing}\ket{e_{x_1}}=0, \quad \mathrm{and} \quad \bra{e_{x_1}} \lr{\tcM_{\varnothing}+\tcM^\dagger_{\varnothing}}\ket{e_{x_1}}=0.
    \end{equation}
    This is possible because two traceless Hermitian matrices can be simultaneously zero-diagonalized \cite{Zhou_LOCC20}: there is always a basis such that one matrix is proportional to Pauli $X$, while the other is linear combination of $X$ and $Y$. Moreover, ${\rm Im}\bra{e_{0}} \tcM_{\varnothing}\ket{e_{0}}$ is chosen to be positive. This completely determines $e_{x_1}$, as long as $\cM_{\varnothing}$ and $\tcM_{\varnothing}$ are nonvanishing.
    \item Calculate $\cM_{x_1}=\mathrm{tr}_{3\cdots N}\bra{e_{x_1}} \cM\ket{e_{x_1}}, \tcM_{x_1}=\mathrm{tr}_{3\cdots N}\bra{e_{x_1}} \tcM\ket{e_{x_1}}$.  We are using $\langle e_{x_1}|\mathcal{M}|e_{x_1}\rangle$ to denote an operator acting only on spins $2,\ldots,N$.
    \item Find an orthonormal basis $\{\ket{e_{x_1x_2}}:x_2=0,1\}$ for spin $2$ s.t. \begin{equation}\label{eq:finde2}
        \bra{e_{x_1x_2}} \cM_{x_1}\ket{e_{x_1x_2}}=\bra{e_{x_1x_2}} \lr{\tcM_{x_1}+\tcM_{x_1}^\dagger}\ket{e_{x_1x_2}}=0.
    \end{equation}  
    The sign of ${\rm Im}\bra{e_{x_10}} \tcM_{x_1}\ket{e_{x_10}}$ is chosen to be $(-1)^{x_1}$.
    \item Repeat steps 3 and 4 for spins $3,\cdots,N$.
\end{enumerate}
One can verify that condition \eqref{eq:1=-x0} is satisfied in the final step.
\end{proc}

\subsection{Matrix product state approximations}

\begin{figure}[t]
\centering
\includegraphics[width=.9\textwidth]{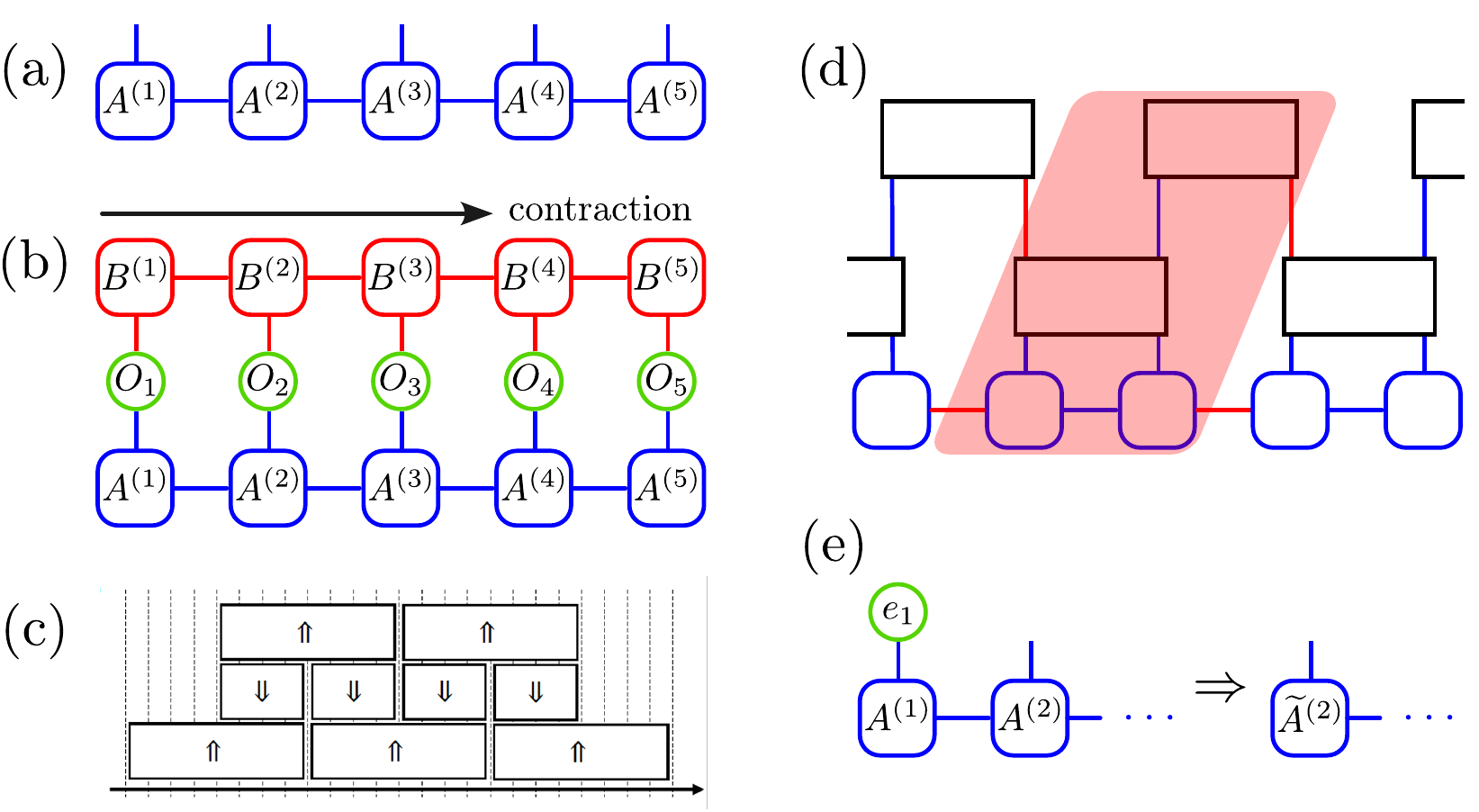}
\caption{\label{fig:MPS} (a) Tensor representation of MPS $\ket{\bf{A}}$ for $N=5$. (b) Tensor representation for $\braket{\mathbf{B}|O|\mathbf{A}}$, where $O=O_1\otimes O_2\otimes \cdots$ is a string operator. It can be contracted along the spatial direction in $\poly{N}$ memory and runtime, if the bond dimensions of the two MPSs are $\poly{N}$. More precisely, at the $n$-th step one has contracted the tensors on the leftmost $n$ qubits to a matrix $C^{(n)}$ of dimension $\poly{N}$; $C^{(n)}$ is then contracted with tensors $A^{(n+1)}, O_{n+1}, B^{(n+1)}$ to yield $C^{(n+1)}$, which costs $\poly{N}$ overhead. When finding the measurement basis for the LOCC protocol, we need e.g. $O=e_{x_1}\otimes \cdots \otimes e_{x_1\cdots x_n}\otimes O_{n+1}$ to calculate $\tcM_{x_1\cdots x_n}$. (c) $3$-layer circuit unitary $U'$ that approximates $\ee^{-\ii t H_{\omega'}}$ with error $\epsilon$. Each gate acts on at most $L=\order(\log(tN/\epsilon) )$ neighboring spins, so $U'$ maps MPS to MPS preserving the condition that the bond dimension is $\poly{N}$. Adapted from \cite{HHKL}. (d) Acting a $2$-layer unitary circuit on an MPS, the bond dimension is increased at most by a factor of the local physical dimension. This can be seen by viewing the shaded region as the new matrix for the final state, whose legs are shown by red. (e) A sketch of why MPS is easy to sample: Whenever a qubit is measured, that qubit can be contracted and the neighboring matrix gets updated. The state is still an MPS on one fewer qubit, with the bond dimension unchanged. }
\end{figure}

In general, the states $\ket{\phi^\alpha_{\omega'}}$ and therefore the matrices $\cM,\tcM$ cannot be computed exactly. Instead, we need approximations of them that (\emph{i}) can be computed efficiently, and (\emph{ii}) induces negligible error for quantum metrology. This will lead to approximation $E_{\bm{x}}^{\rm ap}$ of the local basis, along which the quantum measurement is actually done. In this section, we focus on the easier case of one dimension, where we use MPS to approximate states.

An MPS on a chain of qubits $\{1,\cdots,N\}$ is defined by a set of $D\times D$ matrices $\mathbf{A}=\{A^{(j)}_\alpha\}$: \begin{equation}
    \ket{\mathbf{A} } = \sum_{\alpha_1=0,1}\cdots \sum_{\alpha_N=0,1} \tr{A_{\alpha_1}^{(1)} \cdots A_{\alpha_N}^{(N)} } \ket{\alpha_1\cdots \alpha_N},
\end{equation}
where $D$ is called the bond dimension. Pictorially $\{A^{(j)}_\alpha\}$ are tensors with three legs: one physical leg and two bond legs, and the MPS is represented by contracting the bond legs of nearest-neighbors. See Fig.~\ref{fig:MPS}(a). Although any state can be expressed as an MPS, the bond dimension is typically exponentially large in $N$. In contrast, MPSs of physical interest are those whose bond dimension is upper bounded by $\poly{N}$, which contains only a polynomial number of parameters. Such an MPS is easy to manipulate, because many properties of it can be computed classically in polynomial time by contracting the tensors. For example, given a string operator $O=O_1\otimes O_2\otimes \cdots$, its matrix element between two MPSs is represented by sandwiching the operator by the tensors of the two MPSs (one is conjugated). As shown in Fig.~\ref{fig:MPS}(b), one can contract this tensor from left to right, using only $\poly{N}$ memory and runtime if bond dimensions are bounded by $\poly{N}$.  Our first result shows  that given our assumptions in the previous section, such an efficient MPS must exist.

\begin{prop}\label{prop:MPS}
Assuming \eqref{eq:initial_cor} holds for  $\ket{\alpha}_{\rm in}$ in one dimension,
there is an approximation $\ket{\Psi_{\omega'}}=\frac{1}{\sqrt{2}} \lr{\ket{\Phi^0_{\omega'}} + \ket{\Phi^1_{\omega'}}}$ for the state $\ket{\psi_{\omega'} }$,
such that \begin{equation}\label{eq:phi-phi<e}
    \norm{\ket{\Phi^\alpha_{\omega'}} - \ket{\phi^\alpha_{\omega'}}} \le 1/\sqrt{N},
\end{equation}
and $\ket{\Phi^\alpha_{\omega'}}$ is an MPS with $\poly{N}$ bond dimension.
\end{prop}

\begin{proof}
In 1d, any $\ket{\alpha}_{\rm in}$ with exponential decay of correlation \eqref{eq:initial_cor} can be approximated by an MPS $\ket{\alpha}_{\rm IN}$ with bond dimension $\poly{N}$, with a small error \cite{corDecay_area13}: \begin{equation}\label{eq:a'-a}
    \norm{\ket{\alpha}_{\rm IN} -\ket{\alpha}_{\rm in}}\le N^{-3/2}.
\end{equation} 

On the other hand, there exists a quantum circuit unitary $U'$ (which can be found using the HHKL algorithm \cite{HHKL}) shown in Fig.~\ref{fig:MPS}(c) that approximates the evolution unitary $\ee^{-\ii t H_{\omega'}}$ to error \begin{equation}\label{eq:HHKL}
    \norm{U'-\ee^{-\ii t H_{\omega'}}} \le \epsilon,
\end{equation}
where $U'$ consists of three layers of local gates. Each gate is acting on at most $L=\order(\log(tN/\epsilon) )$ neighboring spins, and is simply $\ee^{\mp \ii t H_{\omega', {\rm local}}}$, where $H_{\omega', {\rm local}}$ only contains terms in $H_{\omega'}$ that are supported inside the given region. Only the second layer is backward time evolution with $+$ sign on the exponent. We choose $\epsilon= 1/\sqrt{2N}$ and focus on constant time, so that $L=\order(\log N)$.\footnote{Note that \cite{HHKL} further splits $U'$ to ${\rm polylog}(N/\epsilon)$ layers of universal gate set, which is not needed here.}

We then define $\ket{\Phi^\alpha_{\omega'}}=U'\ket{\alpha}_{\rm IN}$, which satisfies \eqref{eq:phi-phi<e} from \eqref{eq:a'-a} and \eqref{eq:HHKL}.
Since $\ket{\alpha}_{\rm IN}$ is an MPS with $\poly{N}$ bond dimension,
the output $\ket{\Phi^\alpha_{\omega'}}$ from acting quantum circuit $U'$ on it remains to be such an MPS, because $U'$ increases the bond dimension at most by a multiplicative factor $\poly{N}$. To see this, one can combine each block of $L$ qubits to one qudit of dimension $2^L=\poly{N}$; the matrices $A^{(j)}$ for $\ket{\alpha}_{\rm IN}$ are simply multiplied within each block, which does not change the bond dimension. On the other hand, the quantum circuit $U'$ is simply two layers of nearest-neighbor 2-local gates on the qudits (with a layer of 1-local gates in between), which can multiply at most a factor of $\poly{N}$ to the bond dimension, because the quantum circuit can be expressed as a matrix product operator of $\poly{N}$ bond dimension \cite{MPU18}. One can also see this fact in Fig.~\ref{fig:MPS}(d) by identifying the MPS structure of the final state. 
\end{proof}

Based on the MPS approximation, the measurement basis is not exactly the one determined by \eqref{eq:1=-x0}, but an approximate version $E_{\bm x}^{\rm ap}$ that satisfies
\begin{equation}\label{eq:1=-x0_MPS}
    \braket{E_{\bm x}^{\rm ap}| \Phi^1_{\omega'}} = (-1)^{\bm x}\ii \braket{E_{\bm x}^{\rm ap}| \Phi^0_{\omega'}}.
\end{equation}
This basis is again determined by Procedure \ref{proc:Ex}; the only difference is that the matrices are substituted by 
\begin{subequations}\label{eq:Map}
\begin{align}
    \cM^{\rm ap} &= \ket{\Phi^0_{\omega'}}\bra{\Phi^0_{\omega'}}-\ket{\Phi^1_{\omega'}}\bra{\Phi^1_{\omega'}},\\
    \tcM^{\rm ap} &= \ket{\Phi^1_{\omega'}}\bra{\Phi^0_{\omega'}}.
\end{align}
\end{subequations}
This procedure becomes efficient to implement, as we will show.

\subsection{Efficient metrology protocol with Heisenberg limit}

With the developments above, we propose the following metrology protocol sketched in Fig.~\ref{fig:result}(b).
\begin{proc}\label{proc:metro} Robust metrology protocol with LOCC measurement in 1d: 
    \begin{enumerate}
        \item Prepare the quantum system in state $\ket{\psi}_{\rm in}$ in \eqref{eq:psi_in} that satisfies \eqref{eq:Z0-Z1=cin} and \eqref{eq:initial_cor}.
        \item Evolve under $H_{\omega}$ for time $t$.
        \item Measure $\{E_{\bm x}^{\rm ap}\}$ determined by Procedure \ref{proc:Ex} with MPS-approximated matrices \eqref{eq:Map}, record the parity $\pm 1$ of the string $x_1\cdots x_N$. In this step, do not calculate $E_{\bm x}^{\rm ap}$ for all $\bm x$ because there are exponentially many of them. Instead, just calculate the one chosen by the measurement trajectory. 
        
        \item Repeat steps 1-3 for $1\ll M\ll N$ times to extract an estimate of the average parity \begin{equation}
            \braket{\bf P} = \sum_{\bm x} (-1)^{\bm x} p_{\bm x}, \where p_{\bm x}:= \bra{\psi_{\omega}} E_{\bm x}^{\rm ap} \ket{\psi_{\omega}}.
        \end{equation}

        \item Classically simulate \begin{equation}
            \braket{\bf P}' = \sum_{\bm x} (-1)^{\bm x} p_{\bm{x}}',\where p_{\bm{x}}':=\bra{\psi_{\omega'}} E_{\bm x}^{\rm ap} \ket{\psi_{\omega'}},
        \end{equation}
        with error bounded by $1/\sqrt{N}$. Note that this quantity does not depend on the unknown $\omega$.

        \item Match the experimental result of $\braket{\bf P}-\braket{\bf P}'$ with $-\sin[f(\omega)]$ to read out the parameter $\omega$, using binary search in Proposition \ref{prop:binary}, where the function $f(\tilde{\omega})$ is classically computed for any given $\tilde{\omega}$ with error bounded by $1/\sqrt{N}$.
    \end{enumerate}
\end{proc}

Clearly, this protocol is to globally estimate any $\omega\in \mathcal{I}_{\omega'}$, and involves $\poly{N}$ quantum resources.
We show that for short-time dynamics, it achieves HL and involves only $\poly{N}$ classical computation.

\begin{thm}\label{thm:metro1d}
For $t<c_{\rm in}/J$ in 1d, Procedure \ref{proc:metro}

(0) gives an unbiased estimation of $\omega$ at large $N$: \begin{equation}\label{eq:P-P=f}
    \braket{\bf P} - \braket{\bf P}'= -\sin[f(\omega)] + \order(1/\sqrt{N}),
\end{equation}
and satisfies HL.

Furthermore, the classical computation involved has $\poly{N}$ runtime. Specifically, there is a $\poly{N}$ classical algorithm for calculating each of the following up to error $\order(1/\sqrt{N})$:

(1) $f(\tilde{\omega})$ for any $\tilde{\omega}\in \mathcal{I}_{\omega'}$.

(2) The local basis by Procedure \ref{proc:Ex} for a given trajectory $\bm x$.

(3) $\braket{\bf P}'$.

\end{thm}

\begin{proof}
We prove the claims (0-3) one by one.

(0) $\braket{\bf P}$ is the expectation of observable ${\bf P}=\sum_{\bm x}(-1)^{\bm x} E_{\bm x}^{\rm ap}$ in the final state $\psi_\omega$ with $\norm{\bf P}=1$. \eqref{eq:psi'_omega} together with \eqref{eq:phi-phi<e} implies \begin{equation}\label{eq:psiw=Psiw}
    \ket{\psi_\omega} = \ket{\Psi_\omega} + \order(1/\sqrt{N}), \where \ket{\Psi_\omega} =\frac{1}{\sqrt{2}} \lr{\ee^{-\ii f(\omega)/2} \ket{\Phi^0_{\omega'}} + \ee^{\ii f(\omega)/2} \ket{\Phi^1_{\omega'}} },
\end{equation}
so \begin{subequations} \label{eq:P=P+1_N}
\begin{align}
    \braket{\bf P} &= \bra{\Psi_\omega} {\bf P} \ket{\Psi_\omega} + \order(1/\sqrt{N}), \\
    \braket{\bf P}' &= \bra{\Psi_{\omega'}} {\bf P} \ket{\Psi_{\omega'}} + \order(1/\sqrt{N}).
\end{align} \end{subequations}
By explicit calculation, \begin{align}\label{eq:P=P-f_MPS}
    \bra{\Psi_\omega} {\bf P} \ket{\Psi_\omega} &= \sum_{\bm x}(-1)^{\bm x} \glr{ \frac{1}{2}\lr{\braket{E_{\bm x}^{\rm ap}}_{\Phi^0_{\omega'}} + \braket{E_{\bm x}^{\rm ap}}_{\Phi^1_{\omega'}} } + \mathrm{Re} \mlr{\ee^{\ii f(\omega)} \bra{\Phi^0_{\omega'}} E_{\bm x}^{\rm ap} \ket{\Phi^1_{\omega'}} } } \nonumber\\
    &= \bra{\Psi_{\omega'}} {\bf P} \ket{\Psi_{\omega'}} - \sum_{\bm x} \abs{\bra{\Phi^0_{\omega'}} E_{\bm x}^{\rm ap} \ket{\Phi^1_{\omega'}}} \sin[f(\omega)] \nonumber\\
    &= \bra{\Psi_{\omega'}} {\bf P} \ket{\Psi_{\omega'}} - \sum_{\bm x} \frac{1}{2}\lr{\braket{E_{\bm x}^{\rm ap}}_{\Phi^0_{\omega'}} + \braket{E_{\bm x}^{\rm ap}}_{\Phi^1_{\omega'}} } \sin[f(\omega)]\nonumber\\
    &= \bra{\Psi_{\omega'}} {\bf P} \ket{\Psi_{\omega'}} -\sin[f(\omega)].
\end{align}
Here the first line follows from \eqref{eq:P'x=} and \eqref{eq:psiw=Psiw}. In the third line, we have used the MPS version of \eqref{eq:1=-x0}. To get the last line, we used $\sum_{\bm x} E_{\bm x}^{\rm ap}=1$. Thus \eqref{eq:P-P=f} is established by combining \eqref{eq:P=P+1_N} and \eqref{eq:P=P-f_MPS}.

Since $M\ll N$, the experimental error $1/\sqrt{M}$ for $\alr{\bf P}$ is dominant in \eqref{eq:P-P=f}. As a result, we can ignore the $1/\sqrt{N}$ errors from the final term, or the computation of $\alr{\bf P}'$ and $f(\omega)$. The HL is therefore satisfied, because the experimental $\alr{\bf P}$ is compared to a function of $\omega$ that has $\propto N$ slope according to \eqref{eq:dfw_N}, using the binary search algorithm in Proposition \ref{prop:binary}.

(1) From \eqref{eq:psi'_omega} and \eqref{eq:cZ=Z},
\begin{align}\label{eq:fw=Z}
    f(\tilde{\omega}) &= \ii \int^{\tilde{\omega}}_{\omega'} \dd \omega''\, \bra{\phi^0_{\omega''}}\partial_{\omega''}\ket{\phi^0_{\omega''}} - \bra{\phi^1_{\omega''}}\partial_{\omega''}\ket{\phi^1_{\omega''})} \nonumber\\ 
    &= \int^{\tilde{\omega}}_{\omega'} \dd \omega''\, \int^t_0 \dd s\, \braket{\ee^{\ii s H_{\omega''}}Z\ee^{-\ii s H_{\omega''} }}_0- \braket{\ee^{\ii s H_{\omega''}}Z\ee^{-\ii s H_{\omega''} }}_1 \nonumber\\
    &\approx \Delta_{\omega}\Delta_t \sum_{\omega''=\omega',\omega'+\Delta_\omega,\cdots,\tilde{\omega}}\sum_{s=0,\Delta_t,\cdots,t}  \braket{\ee^{\ii s H_{\omega''}}Z\ee^{-\ii s H_{\omega''} }}_0- \braket{\ee^{\ii s H_{\omega''}}Z\ee^{-\ii s H_{\omega''} }}_1,
\end{align}
\red{where we approximate the integral by a sum in the last line with step-sizes $\Delta_\omega,\Delta_t$ to be determined. \eqref{eq:fw=Z}} reduces to expectation of local $Z_j$ in finite time evolution from $\ket{\alpha}_{\rm in}$. According to Lieb-Robinson bound \eqref{eq:LRB}, to simulate $\braket{\ee^{\ii s H_{\omega''}}Z_j\ee^{-\ii s H_{\omega''} }}_0$ with error $\order(N^{-1/2})$,\footnote{\red{This scaling is chosen to ensure that the overall error in $f(\tilde\omega)$ is $\mathrm{O}(N^{-1/2})$.  To ensure this, note that $\sum_{\omega^{\prime\prime}} \Delta_\omega = |\tilde\omega - \omega^\prime| = \mathrm{O}(1/N)$, while we have $N$ on site terms in $Z$ and thus the difference of $\langle \cdots\rangle_0 - \langle \cdots \rangle_1$ will scale as $\mathrm{O}(N)$.  We therefore need the accuracy in simulating $\braket{\ee^{\ii s H_{\omega''}}Z_j\ee^{-\ii s H_{\omega''} }}_0$ to be $\mathrm{O}(N^{-1/2})$.}} one can truncate the dynamics in a region of size $\order(\log N)$ around $j$. Furthermore, $\ket{\alpha}_{\rm in}$ can be replaced by MPS $\ket{\alpha}_{\rm IN}$, because the error \eqref{eq:a'-a} contributes $\order\lr{1/\sqrt{N}}$ to \eqref{eq:fw=Z}. It takes $\poly{N}$ classical runtime to contract the matrices to get the initial density matrix of $\ket{\alpha}_{\rm IN}$ on the $\order(\log N)$-size region. Evolving the state and taking expectation of $Z_j$ invokes another $\poly{N}$ overhead. So the total complexity for computing \red{the last line of \eqref{eq:fw=Z} is $\poly{N}/\lr{\Delta_\omega\Delta_t}$, because there are $\order(1/\lr{N\Delta_\omega\Delta_t})$ number of summation terms. Finally, since each summation term varies with both $\omega''$ and $s$ by a bounded slope $\order(N^2)$ (e.g. $\abs{\partial_s \braket{\ee^{\ii s H_{\omega''}}Z\ee^{-\ii s H_{\omega''}}}_0 }\le \norm{[H_{\omega''}, Z]} = \order(N^2)$), the error of approximating the integral is bounded by $\order(1/\lr{N\Delta_\omega\Delta_t}\times N^2\Delta_\omega^2\times N^2\Delta_t^2)= \order(N^3\Delta_\omega\Delta_t)=\order(N^{-1})$ by choosing $\Delta_\omega,\Delta_t=\order(N^{-2})$. To summarize, $f(\tilde{\omega})$ can be computed with error $\order(N^{-1/2})$ in $\poly{N}$ time.}  

(2) Steps 2 and 4 in Procedure \ref{proc:Ex} are just calculations for $2$-by-$2$ matrices, so it suffices to show $\poly{N}$ complexity for steps 1 and 3, where in general one wants to compute \begin{equation}\label{eq:Mxn+1}
    \cM_{x_1\cdots x_{n} }^{\rm ap}=\mathrm{tr}_{n+2\cdots N}\bra{e_{x_1}^{\rm ap}}\cdots \bra{e_{x_1\cdots x_n}^{\rm ap}} \cM^{\rm ap} \ket{e_{x_1}^{\rm ap}}\cdots \ket{e_{x_1\cdots x_n}^{\rm ap}},
\end{equation}
and $\tcM_{x_1\cdots x_{n} }^{\rm ap}$ defined in a similar way. Since $\cM^{\rm ap}$ is the difference of (the density matrices of) two MPSs \eqref{eq:Map}, \eqref{eq:Mxn+1} can be computed by contracting tensors as shown in Fig.~\ref{fig:MPS}(b), which takes $\poly{N}$ classical runtime due to $\poly{N}$ bond dimension. This contraction algorithm works for $\tcM^{\rm ap}$ similarly.  
An alternative way to see this is depicted in Fig.~\ref{fig:MPS}(e): each time a qubit gets projected, the remaining qubits keep in an MPS with bond dimensions unchanged, so that one can easily iterate this process.

(3) Since the next measurement basis is always easily calculable as shown in (2), sampling string $\bm{x}$ with probability $p'_{\bm{x}}$ is classically easy, because any adaptive single-site measurement (e.g. measurement-based quantum computation) on MPS with $\poly{N}$ bond dimension is classically simulable \cite{tree_MPS06,MPS_rep07}. The idea is that one can sample the classical distribution $p'_{\bm x}$ by outputing the bits $x_1,x_2,\cdots$ one by one, where $x_1$, for example, is generated by a random variable $0,1$ with probabilities $\braket{e_{0}^{\rm ap}}_{\Psi_{\omega'}}, \braket{e_{1}^{\rm ap}}_{\Psi_{\omega'}}$ accordingly. These two probabilities are equal to $\braket{e_{0}^{\rm ap}}_{\Phi^0_{\omega'}}, \braket{e_{1}^{\rm ap}}_{\Phi^0_{\omega'}}$ due to the MPS version of \eqref{eq:P'x=} and \eqref{eq:Map}, which are easy to calculate from contracting tensors of the MPS $\Phi^0_{\omega'}$ described in (2). 
After sampling $N$ strings $\bm{x}$, one gets an estimate for $\braket{\bf P}'$ with error $\sim 1/\sqrt{N}$. This in total takes $\poly{N}$ classical resources.
\end{proof}

\section{Sampling of real time evolution}\label{sec:samp}
The MPS techniques in the previous section do not generalize to higher dimensions. For example, we used a $\poly{N}$ classical algorithm that samples LOCC measurement on MPS; while 2d tensor network states like the cluster state are hard to sample classically, because they have the potential to perform universal quantum computation. To make progress, in this section we restrict the state to be obtainable by short-time evolution with a local Hamiltonian from a product state.  For such states, we give a $\poly{N}$ sampling algorithm under certain conditions. This section can be read independently of the rest of the paper, which will subsequently apply the results of this section to develop an efficient LOCC metrology protocol in a general interaction graph $G$. 

We consider a time-dependent Hamiltonian \begin{equation}\label{eq:H=Ha}
    H(s) = \sum_a \lambda_a J_a(s) H_a,
\end{equation}
acting on qubits $\Lambda=\{1,2,\cdots,N\}$,
where each $H_a$ is a local Pauli string with a time-dependent coefficient $\lambda_a J_a(s)$ satisfying $|J_a(s)|\le \tilde{J}$, and $|\lambda_a|\le 1$. \footnote{Note that we use $\tilde{J}$ to differ from the local strength $J$ defined in \eqref{eq:J=}, although they can be related by a constant determined by geometry.}
We further require the whole Hamiltonian has $F$ Fourier components $\{\omega_\nu\}$, i.e., \begin{equation}\label{eq:Jas=}
    J_a(s) = \sum_{\nu=1}^F \ee^{-\ii \omega_\nu s} J_{a,\nu},
\end{equation}
where $\omega_\nu$ does not depend on $a$. We say two Pauli strings $H_a$ and $H_b$ ($a\neq b$) overlap if there is a qubit on which they both act nontrivially. Let $\fkd$ be the maximal number of $H_b$s that an $H_a$ overlaps with, and $k$ be the maximal number of qubits an $H_a$ acts on (i.e., $H$ is $k$-local). 

The system starts from initial product state $\ket{0\cdots 0}$, evolves under $H(s)$ for time $t$, and is then measured by LOCC in local basis $E_{\bm{x}}$ introduced in \eqref{eq:Ex_POVM}. This process generates a binary string $\bm{x}:=(x_1,\cdots, x_N) \in \{0,1\}^N$ with probability distribution \begin{equation}
    p_{\bm x} = \braket{\phi|E_{\bm{x}}|\phi}, \where \ket{\phi}=\mathcal{T} \ee^{-\ii \int^t_0 \dd s H(s)} \ket{0\cdots 0},
\end{equation}
where $\mathcal{T}$ means time-ordering.
We will show that, under certain conditions, $p_{\bm{x}}$ is easy to sample on a classical computer \red{if $\tilde{J}t$ is smaller than a constant}. In particular, the frequencies $\omega_\nu$s can be arbitrarily large comparing to $\tilde{J}$. We focus on qubit systems, but the results generalize naturally to higher-dimensional local degree-of-freedom, i.e., qudits. 

\red{In Section \ref{sec:intera_pic}, we show how the results about the time-dependent case \eqref{eq:H=Ha} can be modified to the time-independent case \eqref{eq:Hw} with an extra local field $\omega Z$, which we no longer assume needs to be small.}
    
\subsection{Assumption on the local measurement basis}
Our rigorous results will hold for the following situation.
Suppose the local measurement basis always have a finite overlap with the local initial state $\ket{0}$:  
\begin{equation}\label{eq:overlap>cm}
    \abs{ \alr{0| e_{x_1\cdots x_n}} } \ge c_{\rm m},\quad \forall \text{ substring } x_1\cdots x_n,
\end{equation}
for some constant $0<c_{\rm m} \le 1/\sqrt{2} $ (subscript $\rm m$ means measurement). The upper bound $1/\sqrt{2}$ has to be there in order to satisfy \eqref{eq:overlap>cm} for $ \ket{e_{x_1\cdots x_{n-1}0}}$ and $\ket{e_{x_1\cdots x_{n-1}1}}$ that comprise a local orthonormal basis. Due to the same reason, \eqref{eq:overlap>cm} also implies \begin{equation}\label{eq:1overlap>cm}
    \abs{ \alr{1| e_{x_1\cdots x_n}} }^2 = 1-\abs{ \alr{0| e_{x_1\cdots x_n} }}^2 = \abs{ \alr{0| e_{x_1\cdots x_{n-1} (1-x_n)} }}^2 \ge c_{\rm m}^2.
\end{equation}
Therefore condition \eqref{eq:overlap>cm} merely says the measurement basis is far from the computational basis (that contains the initial state). This is satisfied by, for example, measurement in the $x$-basis.

Intuitively, to do useful quantum information processing in short time, one typically wants to measure in a way like \eqref{eq:overlap>cm}, instead of in the computational basis $\{\ket{0}, \ket{1}\}$. The reason is that after short time evolution, the qubit is still close to its initial state $\ket{0}$, so the computational-basis measurement would just output a boring $0$ with high probability. Indeed, we will see in the next section that for robust metrology, one always measures in a basis satisfying \eqref{eq:overlap>cm}, like the ideal case with $x$-basis measurement.


\subsection{Polynomial-time classical sampling of short-time evolved product states}

\begin{thm}\label{thm:samp}
    For any LOCC measurement satisfying \eqref{eq:overlap>cm}, there is a constant time \begin{equation}\label{eq:tstar}
        t_*= \mlr{4\ee^2c_{\rm m}^{-2k}\fkd (\fkd+1)\tilde{J}}^{-1},
    \end{equation}
    such that for all $t<t_*$, the measurement string $\bm{x}$ can be classically sampled in $\poly{N}$ runtime, where the output probability distribution $p^{\rm cl}_{\bm{x}}$ ($\rm cl$ means ``classical'') is close to the true one $p_{\bm{x}}$ in total variation distance: \begin{equation}\label{eq:Pcl-P<}
        \norm{p^{\rm cl}-p}_1:= \sum_{\bm{x}} \abs{p^{\rm cl}_{\bm{x}} - p_{\bm{x}} } \le 1/\poly{N}.
    \end{equation}
\end{thm}

As will be shown in the proof, it is crucial that the \emph{whole} Hamiltonian has a bounded number $F$ of frequencies; if each local term has its own set of $F$ frequencies that do not agree with those of other terms, one can adjust the proof to achieve a quasi-polynomial $\order(N^{\log\log N})$ runtime algorithm. 

\begin{proof}
The proof closely follows \cite{our_highT}, where we showed the imaginary-time version of Theorem \ref{thm:samp}, namely sampling high-temperature thermal states is easy (see also \cite{zero_free22}). 

As the starting point, sampling the whole string $\bm{x}$ is reduced to the problem of computing conditional probabilities (marginals) $p(x_{n+1}|\bm{x}_n)$ for measurement outcome $x_{n+1}$ to appear after a previous substring $\bm{x}_n:=(x_1,\cdots,x_n)$. More precisely, the sampling algorithm works as follows: First choose $x_1\in \{0,1\}$ from probability $p(x_1)$, then choose $x_2$ based on marginal $p(x_2|x_1)$, then choose $x_3$ based on $p(x_3|x_1x_2)$, and so on. Lemma 4 in \cite{our_highT} guarantees that to achieve \eqref{eq:Pcl-P<} with the right hand side being $N^{1-\alpha}$, it suffices to get an approximation $p^{\rm cl}(x_{n+1}=0|\bm{x}_n)$ with error ($p^{\rm cl}(1|\bm{x}_n)$ will be just $1-p^{\rm cl}(0|\bm{x}_n)$) \begin{equation}\label{eq:marg-marg<}
    \abs{p^{\rm cl}(0|\bm{x}_n) - p(0|\bm{x}_n)} \le N^{-\alpha}, \where \alpha>1.
\end{equation}

For each marginal, we invoke the following Lemma that we prove later.
\begin{lem}\label{lem:clust}
    If $t<t_*$ in \eqref{eq:tstar} and \eqref{eq:overlap>cm} holds, the marginal is given by
    \begin{equation}\label{eq:p=expan}
        p(0|\bm{x}_n) = \abs{\braket{0| e_{x_1\cdots x_n0}} }^2+ \sum_{m=1}^\infty \gamma_m t^m,
    \end{equation}
    where $\gamma_m$ is bounded: \begin{equation}\label{eq:gamma<beta}
        \abs{\gamma_m} \le m t_*^{-m}.
    \end{equation}
    This implies the series \eqref{eq:p=expan} converges absolutely. Moreover, $\gamma_m$ can be computed in time 
    $\order\mlr{\exp(cm)}$ where \begin{equation} \label{eq:c=logF}
        c=\log\lr{ \fkd 2^{2k+1} F},
    \end{equation}
    is a constant; recall the definitions of $\fkd$ and $k$ below \eqref{eq:Jas=}.
\end{lem}

This Lemma implies a polynomial algorithm for computing $p^{\rm cl}(0|\bm{x}_n)$ in \eqref{eq:marg-marg<}, by truncating the sum to $m\le M=\eta \log N$, where $\eta > \alpha / \log(t_*/t)$.  Combined with the sampling-to-computing reduction above, we get the desired sampling algorithm, with runtime $\order(N^{1+c\eta})$ to output one string $\bm{x}$.
\end{proof}

\subsection{Proof of Lemma \ref{lem:clust} using cluster expansion, \red{and its corollary}}
We need an operator formalism for the proof.\footnote{We refer to Section 2.1 of \cite{review_us23} for a comprehensive review, although some notations are slightly different from here.} Namely, each operator $O$ is viewed as a vector $\keto{O}$ in an ``operator Hilbert space". For example, the initial density matrix \begin{equation}\label{eq:bm0}
    \keto{\bm{0}} := \ket{0\cdots 0} \bra{0\cdots 0} = \bigotimes_{j=1}^N \ket{0}_j\bra{0},
\end{equation}
is a direct-product vector in the operator space. Define the inner product between operators \begin{equation}\label{eq:O_inner}
    \lr{O|O'}:=\tr{O^\dagger O'},
\end{equation}
which induces the operator $2$-norm \begin{equation}
    \norm{\keto{O}}_2 := \sqrt{\lr{O|O}}.
\end{equation}
We rewrite the (forward-in-time) Heisenberg evolution as \begin{equation}
    \frac{\dd}{\dd s} \keto{O(s)} = \cL(s) \keto{O(s)}, \where \cL(s):= -\ii [H(s),\cdot].
\end{equation}
One can verify that the Liouvillian $\cL(s)$ is anti-Hermitian in terms of inner product \eqref{eq:O_inner}.
As a result, the final density matrix $\ket{\phi}\bra{\phi}$, as an operator $\keto{\phi}$, is \begin{equation}\label{eq:final_oper}
    \keto{\phi} = \mathcal{T} \ee^{\int^t_0 \dd s \cL(s)} \keto{\bm{0}},
\end{equation}
where $\ii \cL(s)$ plays the role of a Hermitian Hamiltonian.

\begin{proof}[Proof of Lemma \ref{lem:clust}]
Using \eqref{eq:final_oper}, the conditional probability satisfies
\begin{equation}\label{eq:p=partial}
    p(0|\bm{x}_n)-\abs{\braket{0| e_{x_1\cdots x_n0}} }^2= \frac{ \brao{\bm{x}_n 0} \mathcal{T} \ee^{\int^t_0 \dd s \cL(s)} \keto{\bm{0}} }{ \brao{\bm{x}_n} \mathcal{T} \ee^{\int^t_0 \dd s \cL(s)} \keto{\bm{0}} } -\abs{\braket{0| e_{x_1\cdots x_n0}} }^2= \partial_\kappa \log \cZ |_{\kappa=0},
\end{equation}
where the ``partition function'' 
\begin{equation}\label{eq:Z=}
    \cZ := \frac{\brao{\bm{x}_n}\ee^{\kappa \mathcal{E}_{n+1} } \mathcal{T} \ee^{\int^t_0 \dd s \cL(s)} \keto{\bm{0}}}{ \lr{\bm{x}_n|\bm{0}}}.
\end{equation}
Note that we extract the term $\abs{\braket{0| e_{x_1\cdots x_n0}} }^2$ in \eqref{eq:p=partial} such that the right hand side vanishes for $\beta=0$, and will become the $m\ge 1$ expansion in \eqref{eq:p=expan}.
In \eqref{eq:Z=}, $\mathcal{E}_{n+1}$ is the superoperator that multiplies $e'_{n+1}$ to the left of an operator -- $\mathcal{E}_{n+1}|O) = |e'_{n+1} O)$ --  where $e'_{n+1}:= e_{x_1\cdots x_n 0}-\abs{\braket{0| e_{x_1\cdots x_n0}} }^2$ acts on qubit $n+1$,
and \begin{equation}
    \keto{\bm{x}_n} := \bigotimes_{j=1}^n e_{x_1\cdots x_j} \bigotimes_{j=n+1}^N I_j, \label{eq:keto_xn}
\end{equation}
is a direct-product operator.

Thanks to our assumption \eqref{eq:overlap>cm}, the denominator in \eqref{eq:Z=} is nonzero. It does not matter that the denominator is exponentially small in $n$, as we will see.

Now compare to Lemma 5 in \cite{our_highT}. The conclusions of the lemmas are almost the same, where the value of $t_*$ \eqref{eq:tstar} has an extra factor $\lr{2c_{\rm m}^{-2k}\tilde{J}}^{-1}$ comparing to $\beta_*$ in \cite{our_highT}, and an extra $F$ appears in \eqref{eq:c=logF} for the runtime. On the other hand, the starting point,
\eqref{eq:p=partial} and \eqref{eq:Z=}, are in a similar form to Eqs. (A2) and (A3) in \cite{our_highT}. The differences are (\emph{1}) \eqref{eq:Z=} is real-time evolution viewing $\ii \cL$ as the Hamiltonian, not imaginary-time $\ee^{-\beta H}$; (\emph{2}) \eqref{eq:Z=} is a ratio of matrix elements (in the operator space), not a trace; (\emph{3}) evolution in \eqref{eq:Z=} is time-dependent. We show that these differences can be incorporated into the proof of Lemma 5 in \cite{our_highT}, which leads exactly to the extra factors in Lemma \ref{lem:clust}. We encourage the reader to get familiar with the proof in \cite{our_highT} first, which relies on the technique of cluster expansion developed in \cite{learn_highT21, Loschmidt_echo23}.

Difference (\emph{1}) is straightforward to tackle, because the proof in \cite{our_highT} does not rely on the fact that the inverse-temperature $\beta$ is real: it can be replaced by $\ii t$ here. Below we show how to deal with (\emph{2}) and (\emph{3}).

The key quantity to be bounded is the cluster derivative of $\cZ$, the simplest case $\kappa=0$ being \begin{equation}\label{eq:DVZ}
    \lr{\mathcal{D}_{\V} \cZ}_0 = \sum_{\sigma\in \mathsf{S}_m } \int_{\mathbb{T}_m} \dd \bm{s} J_{a_{\sigma_1}}(s_1)\cdots J_{a_{\sigma_m}}(s_m) \frac{\brao{\bm{x}_n} \cL_{a_{\sigma_1}}\cdots\cL_{a_{\sigma_m}} \keto{\bm{0}}}{ \lr{\bm{x}_n|\bm{0}}}
\end{equation}
where $\V=\{a_1,\cdots,a_m\}$ is a \emph{connected} cluster of size $m$, $\mathsf{S}_m$ is the $m$-th permutation group, and $\cL_a := -\ii [H_a, \cdot]$. Due to time-ordering, the integral in \eqref{eq:DVZ} \begin{equation}
    \int_{\mathbb{T}_m} \dd \bm{s}:= \int^t_0 \dd s_1 \cdots \int^{s_{m-1}}_0 \dd s_m,
\end{equation}
is on a simplex $\mathbb{T}_m$, whose volume is $\mathrm{Vol}(\mathbb{T}_m)=t^m/m!$. Since the states $\keto{\bm{0}}, \keto{\bm{x}_n}$ are both product states, the matrix element in \eqref{eq:DVZ} can be restricted in the support region $\mathrm{Supp}(\V)\subset \Lambda$ of $\V$: the numerator and denominator factorize (e.g. $\lr{\bm{x}_n | \bm{0} } = \lr{\bm{x}_n | \bm{0} }_{\mathrm{Supp}(\V)}\cdot \lr{\bm{x}_n | \bm{0} }_{\Lambda \backslash \mathrm{Supp}(\V)}$) so that contributions outside the region $\mathrm{Supp}(\V)$ cancel. We then bound the restricted numerator and denominator as follows. The numerator is bounded by Cauchy-Schwarz inequality: \begin{equation}
    \brao{\bm{x}_n} \cL_{a_{\sigma_1}}\cdots\cL_{a_{\sigma_m}} \keto{\bm{0}}_{\mathrm{Supp}(\V)} \le \sqrt{\lr{\bm{x}_n| \bm{x}_n}_{\mathrm{Supp}(\V)}} \lr{\max_a \norm{\cL_a} }^m \sqrt{\lr{\bm{0}| \bm{0}}_{\mathrm{Supp}(\V)}} = 2^{\frac{1}{2}\lr{\abs{\mathrm{Supp}(\V)}-n_{\V}} }\cdot 2^m,
\end{equation}
where we used the superoperator norm bounded by H\"{o}lder inequality $\norm{AB}_2\le \norm{A} \norm{B}_2$: \begin{equation}\label{eq:L<2}
    \norm{\cL_a}:= \max_{\norm{\keto{O}}_2=1} \norm{\cL_a\keto{O}}_2 = \max_{\norm{\keto{O}}_2=1} \norm{\keto{H_aO-OH_a}}_2 \le \max_{\norm{\keto{O}}_2=1} 2 \norm{H_a} \norm{\keto{O}}_2 = 2\norm{H_a} = 2,
\end{equation} 
and $n_{\V}$ denoting the number of measured qubits in $\mathrm{Supp}(\V)$. The denominator is bounded by \eqref{eq:overlap>cm}: \begin{equation}
    \lr{\bm{x}_n | \bm{0} }_{\mathrm{Supp}(\V)} \ge c_{\rm m}^{2n_{\V}}.
\end{equation}
Putting ingredients together,
\eqref{eq:DVZ} is bounded by \begin{align}\label{eq:DVZ<}
    \abs{\lr{\mathcal{D}_\V \cZ}_0} &\le m! \mathrm{Vol}(\mathbb{T}_m) \tilde{J}^m 2^{m+\frac{1}{2}\lr{\abs{\mathrm{Supp}(\V)}-n_{\V}} } c_{\rm m}^{-2n_{\V}}\nonumber\\
    &\le \lr{t\tilde{J}}^m 2^m c_{\rm m}^{-2\abs{\mathrm{Supp}(\V)}} \le \lr{2t\tilde{J}}^m c_{\rm m}^{-2km} =\lr{2c_{\rm m}^{-2k}t\tilde{J}}^m.
\end{align}
Here in the second line, we maximize the expression at the largest $n_{\V}=\abs{\mathrm{Supp}(\V)}\le km$ due to $c_{\rm m}\le 1/\sqrt{2}$. Comparing \eqref{eq:DVZ<} to the bound $\beta^m$ in Eq. (A14) of \cite{our_highT}, we see difference (\emph{2}) just invokes the advertised extra factor $2c_{\rm m}^{-2k}\tilde{J}$. Here $2$ comes from taking commutator in \eqref{eq:L<2}, $c_{\rm m}^{-2k}$ comes from the denominator in \eqref{eq:DVZ}, and $\tilde{J}$ is just fixing the energy unit. \eqref{eq:gamma<beta} then follows from the proof in \cite{our_highT}.

To bound the runtime, observe that difference (\emph{3}) only modifies the first step to calculate a given cluster derivative $\mathcal{D}_\V \log \cZ$. In \cite{our_highT}, the first step is to compute $\mathsf{Z}^{\mathcal{K}}$ (originating from the $\mathcal{K}$-th order $\beta$-expansion of $\ee^{-\beta H}$), where $\mathsf{Z}:=\sum_{a\in \V} z_a H_a$ with $z_a$ being some integers, and $\mathcal{K}\le m$ is a positive integer. Here, due to time dependence, the first step is to compute \begin{equation}
    \mathcal{Q}:=\int_{\mathbb{T}_{\mathcal{K}}} \dd \bm{s}\, \mathsf{Z}(s_1)\cdots \mathsf{Z}(s_{\mathcal{K}}),
\end{equation}
where $\mathsf{Z}(s):=\sum_{a\in \V} z_a J_a(s) H_a$. Thanks to \eqref{eq:Jas=}, $\mathsf{Z}(s)$ only has $F$ frequencies: \begin{equation}
    \mathsf{Z}(s) = \sum_{\nu=1}^F \ee^{-\ii \omega_\nu s} \mathsf{Z}_\nu,
\end{equation}
so that \begin{equation}
    \mathcal{Q} = \sum_{\nu_1 = 1}^F \cdots \sum_{\nu_{\mathcal{K}} = 1}^F \lr{ \int_{\mathbb{T}_{\mathcal{K}}} \dd \bm{s} \ee^{-\ii\sum_{\eta=1}^{\mathcal{K}} \omega_{\nu_\eta} s_\eta} } \mathsf{Z}_{\nu_1}\cdots \mathsf{Z}_{\nu_{\mathcal{K}}}.
\end{equation}
Since the integral inside the bracket can be done analytically, the overhead incurred by time-dependence is merely that instead of a single product of matrices $\mathsf{Z}^{\mathcal{K}}$, one needs to compute $F^{\mathcal{K}}\le F^m$ such products.\footnote{If each coupling $H_a$ has its own set of $F$ frequencies, $\mathsf{Z}(s)$ will have $Fm$ frequencies and $(Fm)^m$ products need to be evaluated. For the final algorithm, this leads to quasi-polynomial overhead $N^{\log \log N}$ because $m=\Theta(\log N)$.} Each product $\mathsf{Z}_{\nu_1}\cdots \mathsf{Z}_{\nu_{\mathcal{K}}}$ is no harder to compute than $\mathsf{Z}^{\mathcal{K}}$: one simply multiplies the $\mathcal{K}$ sparse matrices one by one. \eqref{eq:c=logF} then follows from \cite{our_highT} where one adds an extra factor $F^m$ to the runtime. This finishes the proof of Lemma \ref{lem:clust}.
\end{proof}

\red{
In the next Section, we will also consider the following quantity
\begin{equation}\label{eq:tp=}
    \widetilde{p}(0|\bm{x}_n)
:=\frac{ \brao{\bm{x}_n 0} \mathcal{T} \ee^{\int^t_0 \dd s \cL(s)} \keto{\bm{1,0}} }{ \brao{\bm{x}_n} \mathcal{T} \ee^{\int^t_0 \dd s \cL(s)} \keto{\bm{1,0}} },
\end{equation}
similar to the form of $p(0|\bm{x}_n)$ in \eqref{eq:p=partial}, where \begin{equation}\label{eq:bm10}
    \keto{\bm{1,0}} := \ket{1\cdots 1}\bra{0\cdots 0}.
\end{equation}
Although the operator $\keto{\bm{1,0}}$ is not a quantum state so that $\widetilde{p}(0|\bm{x}_n)$ is not interpreted as a marginal probability, Lemma \ref{lem:clust} generalizes to yield a cluster expansion for $\widetilde{p}(0|\bm{x}_n)$ because of the following: In the above proof, we only need two properties of the initial operator \eqref{eq:bm0}: its direct-product structure, and its nonvanishing local overlap with the previous measurement basis $\keto{\bm{x}_n}$. These properties also hold for $\keto{\bm{1,0}}$, because $\abs{\alr{0| e_{x_1\cdots x_n}} \alr{e_{x_1\cdots x_n}|1}}\ge c_{\rm m}^2$ from \eqref{eq:overlap>cm} and \eqref{eq:1overlap>cm}. As a result, we have the following corollary of Lemma \ref{lem:clust}:
\begin{cor}\label{cor:clust}
    If $t<t_*$ in \eqref{eq:tstar} and \eqref{eq:overlap>cm} holds, 
    \begin{equation}\label{eq:tp=expan}
        \widetilde{p}(0|\bm{x}_n) = \braket{e_{x_1\cdots x_n0}|1} \braket{0| e_{x_1\cdots x_n0}} + \sum_{m=1}^\infty \widetilde{\gamma}_m t^m,
    \end{equation}
    where $\widetilde{\gamma}_m$ has the same bound \eqref{eq:gamma<beta} for $\gamma_m$.
    This implies the series \eqref{eq:tp=expan} converges absolutely. Moreover, $\widetilde{\gamma}_m$ can be computed in the same runtime 
    $\order\mlr{\exp(cm)}$ as $\gamma_m$ where $c$ is defined in \eqref{eq:c=logF}.
\end{cor}
}

\subsection{Measuring short-time evolved states cannot generate long-range correlation}

Similar to Theorem 3 in \cite{our_highT}, a straightforward modification of the proof of Lemma \ref{lem:clust} leads to the following: 
\begin{thm}\label{thm:nocorr}
At short time $t<t_*$ \eqref{eq:tstar}, after projecting $n$ qubits onto local basis using projector $E_{\bm{x}_n}:=\keto{\bm{x}_n}$ \eqref{eq:keto_xn} that satisfies \eqref{eq:overlap>cm}, two unmeasured qubits $i,j$ have correlation exponentially small in their distance $\mathsf{d}(i,j)$
\begin{equation}\label{eq:cor<}
    \mathrm{Cor}(i,j):=\max_{O_i,O_j} \frac{ \bra{\phi}O_iO_j E_{\bm{x}_n} \ket{\phi} }{ \bra{\phi}E_{\bm{x}_n} \ket{\phi}}- \frac{ \bra{\phi}O_i E_{\bm{x}_n} \ket{\phi} }{ \bra{\phi}E_{\bm{x}_n} \ket{\phi}} \frac{ \bra{\phi}O_j E_{\bm{x}_n} \ket{\phi} }{ \bra{\phi}E_{\bm{x}_n} \ket{\phi}}\le c'_{\rm cor} \lr{\frac{t}{t_*}}^{c_{\rm cor}\mathsf{d}(i,j)},
\end{equation}
for some constants $c_{\rm cor},c'_{\rm cor}>0$,
where $O_i,O_j$ act on $i$ and $j$ respectively with $\norm{O_i},\norm{O_j}\le 1$.
\end{thm}

It is well-known that one can perform universal measurement-based quantum computation (MBQC) on the 2d cluster state, which can be prepared from an initial product state by $\Theta(1)$-time Hamiltonian evolution. Theorem \ref{thm:nocorr} suggests that the evolution time has to be $\Omega(1)$ in order for the final state to be a universal MBQC resource. This constraint also holds for large-distance teleportation using measurements \cite{telep_finite_t21}, because teleportation is equivalent to building long-range correlation \cite{Friedman:2022vqb}. It is an interesting question whether the condition \eqref{eq:overlap>cm} can be removed, which may demand deeper understanding of the cluster expansion itself.

\subsection{Eliminating local fields via the interaction picture} \label{sec:intera_pic}
Results in this section generalize to the case \eqref{eq:Hw} where the Hamiltonian $H_\omega$ is time-independent with a local field $\omega Z$. For example, sampling is easy as long as \red{$Jt$ is smaller than a constant}, even if $\omega$ is much larger than the local strength $J$ of $V$. The reason is as follows. In the interaction picture, \begin{equation}
    \ee^{-\ii t H_\omega} =  \lr{\tilde{\mathcal{T}} \ee^{-\ii \int^t_0 \dd s V(s)} }\ee^{-\ii \omega t Z} , \where V(s) = \ee^{-\ii \omega sZ}V\ee^{\ii \omega sZ},
\end{equation}
where $\tilde{\mathcal{T}}$ is anti-time ordering. Acting on the initial state $\ket{0\cdots 0}$, the $\ee^{-\ii \omega t Z}$ part becomes trivial, so the evolution is just generated by a time-dependent Hamiltonian $H(s)=V(t-s)$. Since interaction $V$ has finite range and $Z$ is integer-valued, $V(s)$ has finite number of frequencies each being an integer multiple of $\omega$. Thus the situation reduces to the case \eqref{eq:H=Ha} with \eqref{eq:Jas=}, and we have the following:

\begin{thm}\label{thm:samp1}
    Suppose product state $\ket{0\cdots 0}$ is evolved under Hamiltonian \eqref{eq:Hw} for time $t$, after which LOCC measurement satisfying \eqref{eq:overlap>cm} is performed. There exists a constant $c_+>0$ (independent of $N$) determined by geometry and $c_{\rm m}$, such that for $t<c_+/J$, there is a $\poly{N}$ classical sampling algorithm achieving \eqref{eq:Pcl-P<}.
\end{thm}

\section{LOCC protocol in higher dimensions}\label{sec:>1d}
In Section \ref{sec:1d}, we presented a provably efficient LOCC measurement protocol for Heisenberg-limited metrology in 1d, using the fact that MPS with small bond dimension is easy to sample from. In this section, we establish similar results for general graphs, using the sampling result in Section \ref{sec:samp}. In contrast to 1d, here we deal with a more restricted class of initial states, since correlation decay \eqref{eq:initial_cor} no longer guarantees easiness of sampling.\footnote{Correlation decay is always present in finite-time evolved states \cite{review_us23}, but such states are expected to be hard to sample, since they include cluster states capable of MBQC.} For simplicity, we assume the initial state is just the ideal GHZ state. The following results can be generalized to other initial states, for example, sufficiently short time evolution on $\ket{\GHZ}$, or rotated version of GHZ states.

The metrology protocol is similar to Procedure \ref{proc:metro}. The only difference 
is how to approximate $E_{\bm x}$ determined by \eqref{eq:1=-x0}. In 1d, we used MPS approximations to get $E_{\bm x}^{\rm ap}$; here in the next subsection, we use cluster expansions developed in the previous section (which also applies to 1d as a special case). 
Note that although one still has a tensor network representation for $\psi_{\omega'}$ in higher dimensions \cite{HHKL}, it is not manifestly easy to sample so is not used here. 

\subsection{Finding the next LOCC measurement basis via cluster expansion: assuming no previous errors}

We want to compute an approximation \begin{equation}\label{eq:Eap=e}
    E_{\bm x}^{\rm ap} = e_{x_1}^{\rm ap}\otimes e_{x_1x_2}^{\rm ap}\otimes \cdots\otimes e_{x_1\cdots x_N}^{\rm ap},
\end{equation}
of $E_{\bm x}$ (containing $e_{x_1\cdots x_n}$) that is determined by matrices $\cM$ and $\tcM$ in \eqref{eq:EME=0'} and \eqref{eq:EtME=0}. Following Procedure \ref{proc:Ex}, we compute the local basis in \eqref{eq:Eap=e} one-by-one, where the previously computed $e_{x_1\cdots x_j}^{\rm ap}$s ($j\le n$) are used for the next one $e_{x_1\cdots x_{n+1}}^{\rm ap}$. For illustration purpose, in this subsection we first show that \emph{if the previously computed basis happen to equal the true ones $e_{x_1\cdots x_j}$, and they all satisfy \eqref{eq:overlap>cm},} then the next one can be computed efficiently with small error. We will analyze the accumulated errors, as well as justify \eqref{eq:overlap>cm}, in the next subsection.

Observe that the key quantity to compute in Procedure \ref{proc:Ex} is the $2$-by-$2$ matrix \footnote{From now on we normalize the matrix such that $\cM_{x_1\cdots x_n}=\cM_{x_1\cdots x_n}^{\rm old}/\mathcal{N}_{x_1\cdots x_n}$ where $\cM_{x_1\cdots x_n}^{\rm old}$ is the one used in Section \ref{sec:1d}. Here $\mathcal{N}_{x_1\cdots x_n} := \bra{\phi^\alpha_{\omega'}} e_{x_1}\otimes \cdots \otimes e_{x_1\cdots x_n} \ket{\phi^\alpha_{\omega'}}$ is a normalization factor (which holds for both $\alpha=0,1$ because the previous basis are chosen to zero-diagonalize $\cM_{x_1\cdots}^{\rm old}$ such as \eqref{eq:finde2}). This normalization makes the new matrix $\cM_{x_1\cdots x_n}$ to be of order $1$ (instead of exponentially small for $\cM_{x_1\cdots}^{\rm old}$), and does not change the basis obtained from Procedure \ref{proc:Ex}. $\tcM_{x_1\cdots}$ is normalized analogously.} \begin{equation}
    \cM_{x_1\cdots x_n} =\rho_{x_1\cdots x_{n}}^{0} - \rho_{x_1\cdots x_{n}}^{1},
\end{equation}
where $\rho^0_{x_1\cdots x_n}$ for example is the normalized reduced density matrix of $n+1$ after measuring the previous qubits $1,\cdots,n$, in the time-evolved \emph{all-zero} (instead of GHZ) state $\phi^0_{\omega'}$. One of its matrix elements is given by \begin{equation}\label{eq:M=rho-rho}
    \bra{e_{x_1\cdots x_n0}}\cM_{x_1\cdots x_n} \ket{e_{x_1\cdots x_n0}}= p^0(0|\bm{x}_n) - p^1(0|\bm{x}_n),
\end{equation}
where $p^\alpha(0|\bm{x}_n)$ is just the marginal probability in \eqref{eq:p=expan} that is shown to obey the cluster expansion! Note that \eqref{eq:p=expan} is exactly the case $\alpha=0$, and is straightforwardly generalized to the case $\alpha=1$ where one just changes the initial state to be all-one. Furthermore, since the proof of Lemma \ref{lem:clust} does not depend on the precise basis $\ket{e_{x_1\cdots x_n0}}$, the expectation value of $\rho_{x_1\cdots x_{n}}^\alpha$ on \emph{any} state obeys the cluster expansion. Since the off-diagonal matrix element $\bra{e_{x_1\cdots x_n0}}\rho^\alpha_{x_1\cdots x_n} \ket{e_{x_1\cdots x_n1}}$ is a linear combination of such expectations in states $\ket{e_{x_1\cdots x_n0}}+\ee^{\ii \theta}\ket{e_{x_1\cdots x_n1}}$ where $\theta\in\{0,\pi/2,\pi,3\pi/2\}$, it also obeys a cluster expansion like \eqref{eq:p=expan}. Therefore, according to Lemma \ref{lem:clust}, \emph{all} matrix elements of $\rho^\alpha_{x_1\cdots x_n}$, and furthermore all matrix elements of $\cM_{x_1\cdots x_n}$, can be computed within error $1/\poly{N}$ using $\poly{N}$ classical runtime.

$\tcM_{x_1\cdots x_n}$ can be computed using cluster expansion in a similar way: Explicit calculation similar to \eqref{eq:M=rho-rho} yields \begin{equation}
    \bra{e_{x_1\cdots x_n0}}\tcM_{x_1\cdots x_n}\ket{e_{x_1\cdots x_n0}} =\frac{ \brao{\bm{x}_n 0} \ee^{t \cL_{\omega'}} \keto{\bm{1,0}} }{ \brao{\bm{x}_n} \ee^{t \cL_{\omega'}} \keto{\bm{1,0}} },
\end{equation}
where $\cL_{\omega'}:=-\ii[H_{\omega'},\cdot]$. This exactly corresponds to the quantity $\widetilde{p}(0|\bm{x}_n)$ defined in \eqref{eq:tp=}, which has a cluster expansion of the same form of $p(0|\bm{x}_n)$ in Corollary \ref{cor:clust}. As a result, all matrix elements of $\tcM_{x_1\cdots x_n}$ can also be computed within error $1/\poly{N}$ using $\poly{N}$ classical runtime. 

Due to the proper normalization of $\cM_{x_1\cdots x_n},\tcM_{x_1\cdots x_n}$, the above $1/\poly{N}$ approximation error for the two $2$-by-$2$ matrices (whose largest matrix element is $\Theta(1)$) propagates to $1/\poly{N}$ error on the next computed local basis: 
\begin{equation}
    \norm{e_{x_1\cdots x_{n+1}}^{\rm ap} - e_{x_1\cdots x_{n+1}} } \le 1/\poly{N}, \quad \mathrm{if}\quad e_{x_1\cdots x_j}^{\rm ap} = e_{x_1\cdots x_j}, (j\le n) \text{ and satisfy \eqref{eq:overlap>cm}}.
\end{equation}

\subsection{Finding the LOCC measurement basis: the whole procedure and error analysis}
In reality, $e_{x_1\cdots x_j}^{\rm ap}$ does not equal the ideal $e_{x_1\cdots x_j}$, but one can still perform the computation above using the computed $e_{x_1\cdots x_j}^{\rm ap}$ instead without knowing the ideal one, as long as it obeys \eqref{eq:overlap>cm}. More explicitly, we use the following procedure based on cluster expansion that mimics Procedure \ref{proc:Ex}. We focus on $\cM$, and $\tcM$ is treated analogously.

\begin{proc}\label{proc:Ex_clus} 
Procedure to determine a local basis $e_{x_1}^{\rm ap},e_{x_1x_2}^{\rm ap},\cdots$ along a given trajectory $x_1,x_2,\cdots$ using the cluster expansion algorithm.
\begin{enumerate}
    \item Compute the reduced density matrices $\rho^\alpha_{\varnothing}$ on spin $1$ in the state $\phi^\alpha_{\omega'}$ (whose initial state is all-$\alpha$ instead of GHZ). Using the cluster expansion algorithm above, a polynomial-time computation outputs an approximation $\rho^{\alpha, \mathrm{ap}}_{\varnothing}$ (here subscript $\varnothing$ means no previous substring measured) with $1/\poly{N}$ error for all of its matrix elements. This yields \begin{equation}
        \cM_{\varnothing}^{\rm ap}:= \rho^{0, \mathrm{ap}}_{\varnothing} - \rho^{1, \mathrm{ap}}_{\varnothing},
    \end{equation}
    which approximates $\cM_{\varnothing}$ with $1/\poly{N}$ error. Find an approximation $\tcM_{\varnothing}^{\rm ap}$ for $\tcM_{\varnothing}$ analogously.
    \item Find an orthonormal basis $\ket{e_{x_1}^{\rm ap}}$ for spin $1$ s.t. \begin{equation}\label{eq:E1ME1'=0ap}
        \bra{e_{x_1}^{\rm ap}} \cM_{\varnothing}^{\rm ap}\ket{e_{x_1}^{\rm ap}}=0, \quad \mathrm{and} \quad \bra{e_{x_1}^{\rm ap}} \lr{\tcM_{\varnothing}^{\rm ap}+\tcM^{\rm ap,\dagger}_{\varnothing}}\ket{e_{x_1}^{\rm ap}}=0.
    \end{equation}
    This is possible because two traceless Hermitian matrices can be simultaneously zero-diagonalized \cite{Zhou_LOCC20}. Moreover, ${\rm Im}\bra{e_{0}^{\rm ap}} \tcM_{\varnothing}^{\rm ap}\ket{e_{0}^{\rm ap}}$ is chosen to be positive, which completely determines $e_{x_1}^{\rm ap}$.
    \item Compute the reduced density matrices $\rho^{\alpha}_{x_1, \mathrm{ap}}$ on qubit $2$ after measuring $e^{\rm ap}_{x_1}$ on qubit $1$ in state $\phi^\alpha_{\omega'}$. As long as $e^{\rm ap}_{x_1}$ does not coincide with the computational basis \eqref{eq:overlap>cm}, the computation can be done via cluster expansion in $\poly{N}$ time and outputs an approximation $\rho^{\alpha, \mathrm{ap}}_{x_1}$ of $\rho^{\alpha}_{x_1, \mathrm{ap}}$ (note the conceptual differences between the two) with error $1/\poly{N}$. Use them to compute \begin{equation}
        \cM^{\mathrm{ap}}_{x_1} := \rho^{0, \mathrm{ap}}_{x_1} - \rho^{1, \mathrm{ap}}_{x_1}
    \end{equation}
    as an approximation of $\cM_{x_1}$. Compute $\tcM^{\mathrm{ap}}_{x_1}$ that approximates $\tcM_{x_1}$ similarly.
    \item Find an orthonormal basis $\{\ket{e_{x_1x_2}^{\rm ap}}:x_2=0,1\}$ for spin $2$ s.t. \begin{equation}\label{eq:Map=0}
        \bra{e_{x_1x_2}^{\rm ap}} \cM^{\rm ap}_{x_1}\ket{e_{x_1x_2}^{\rm ap}}=\bra{e_{x_1x_2}^{\rm ap}} \lr{\tcM^{\rm ap}_{x_1}+\tcM_{x_1}^{\rm ap, \dagger}}\ket{e_{x_1x_2}^{\rm ap}}=0.
    \end{equation}  
    The sign of ${\rm Im}\bra{e_{x_10}^{\rm ap}} \tcM_{x_1}^{\rm ap}\ket{e_{x_10}^{\rm ap}}$ is chosen to be $(-1)^{x_1}$.
    \item Repeat steps 3 and 4 for spins $3,\cdots,N$.
\end{enumerate}
\end{proc}

We want to show that Procedure \ref{proc:Ex_clus} is accurate and efficient, in the sense that the measurement basis automatically satisfy \eqref{eq:overlap>cm} enabling the cluster expansion algorithm, and that the errors accumulate in a mild way.

\begin{prop}\label{prop:Ex}
Suppose $\ket{\GHZ}$ is evolved under Hamiltonian \eqref{eq:Hw} for time $t$. Procedure \ref{proc:Ex} defines a LOCC measurement basis $E_{\bm{x}}$ for metrology, while Procedure \ref{proc:Ex_clus} yields an approximate basis $E_{\bm{x}}^{\rm ap}$ based on the cluster expansion algorithm.
There exists positive constants $c_{\rm M}, c_{\rm m}$ (here subscript $\rm M$ stands for ``metrology'') independent of $N$ such that the following holds.
If $t<c_{\rm M}/J$, the local basis elements are far from the computational basis \begin{equation}\label{eq:overlap'>cm}
    \abs{\alr{0| e_{x_1\cdots x_n}^{\rm ap}} } \ge c_{\rm m},\quad \forall \text{ substring } x_1\cdots x_n,
\end{equation}
so that Procedure \ref{proc:Ex_clus} only needs $\poly{N}$-time classical computation.
Furthermore, the result $E_{\bm x}^{\rm ap}$ is close to $E_{\bm x}$ in the sense that \begin{equation}\label{eq:E'-E<}
    \norm{e_{x_1\cdots x_{n+1}}^{\rm ap} - e_{x_1\cdots x_{n+1}} } \le 1/\poly{N}, \quad \forall \text{ substring } x_1\cdots x_n.
\end{equation}
\end{prop}

\begin{proof}
Let $\epsilon=1/\poly{N}$ be the desired accuracy on the right hand side of \eqref{eq:E'-E<}. We have shown in the previous subsection that \eqref{eq:E'-E<} holds for the first step $n=0$ (\eqref{eq:overlap'>cm} also holds trivially), because there is no previously computed basis.
We prove for the later measurements $n=1,2,\cdots$ by induction. Supposing \eqref{eq:overlap'>cm} and \eqref{eq:E'-E<} hold for $0,1,\cdots,n-1$, we want to show that they also hold for $n$.

To show \eqref{eq:overlap'>cm}, we know from \eqref{eq:p=expan} that at sufficiently short time, the conditioned reduced density matrix $\rho^\alpha_{x_1\cdots x_{n-1}, \mathrm{ap}}$ on spin $n$ after evolving $\ket{\alpha \cdots \alpha}$ and measuring $e_{x_1\cdots x_j}^{\rm ap}$ ($j\in\{1,\cdots,n-1\}$) that satisfies \eqref{eq:overlap'>cm} is close to $\ket{\alpha}\bra{\alpha}$. Since the difference, $\rho^0_{x_1\cdots x_{n-1}, \mathrm{ap}} - \rho^1_{x_1\cdots x_{n-1}, \mathrm{ap}}\approx \rho_{x_1\cdots x_{n-1}}^{0, \rm ap} - \rho_{x_1\cdots x_{n-1}}^{1, \rm ap}=\cM_{x_1\cdots x_{n-1}}^{\rm ap}$, have almost zero diagonal elements in basis $\ket{e_{x_1\cdots x_{n}}^{\rm ap}}$, the basis cannot be close to the computational basis. This argument can be shown explicitly by \begin{align}\label{eq:Ecl_far}
    \abs{\braket{0| e_{x_1\cdots x_{n}}^{\rm ap}} }^2 &= \bra{e_{x_1\cdots x_{n}}^{\rm ap}} \rho_{x_1\cdots x_{n-1}, \mathrm{ap}}^{0} \ket{e_{x_1\cdots x_{n}}^{\rm ap}} + \order(Jt) \nonumber\\
    &= \bra{e_{x_1\cdots x_{n}}^{\rm ap}} \rho_{x_1\cdots x_{n-1}}^{0, \mathrm{ap}} \ket{e_{x_1\cdots x_{n}}^{\rm ap}} + \order(\epsilon)+ \order(Jt) \nonumber\\
    &= \bra{e_{x_1\cdots x_{n}}^{\rm ap}} \rho_{x_1\cdots x_{n-1}}^{1, \mathrm{ap}} \ket{e_{x_1\cdots x_{n}}^{\rm ap}} + \order(\epsilon)+ \order(Jt)= \bra{e_{x_1\cdots x_{n}}^{\rm ap}} \rho_{x_1\cdots x_{n-1}, \mathrm{ap}}^{1} \ket{e_{x_1\cdots x_{n}}^{\rm ap}} + \order(\epsilon)+ \order(Jt) \nonumber\\ &= \abs{\braket{1| e_{x_1\cdots x_{n}}^{\rm ap}} }^2+ \order(\epsilon) + \order(Jt) = 1- \abs{\braket{0| e_{x_1\cdots x_{n}}^{\rm ap}} }^2+ \order(\epsilon) + \order(Jt), \nonumber\\
    \rarrow &\abs{\alr{0| e_{x_1\cdots x_{n}}^{\rm ap}} }^2 = 1/2 + \order(\epsilon)+ \order(Jt).
\end{align}
Here in the first line, $\order(Jt)$ comes from the $m\ge 1$ expansion of \eqref{eq:p=expan} and does not depend on $n,N$. To get the second line, we used the induction hypothesis \eqref{eq:overlap'>cm} for the previous measurements that enables a polynomial-time computation for $\rho_{x_1\cdots x_{n-1}, \mathrm{ap}}^{0}$ with error $\order(\epsilon)$. The third line uses \eqref{eq:Map=0}, and the rest uses the above arguments again for $\alpha=1$. For any $c_{\rm m}<1/2$, \eqref{eq:Ecl_far} satisfies \eqref{eq:overlap'>cm} if $c_{\rm M}$ is sufficiently small, because $Jt < c_{\rm M}$ and $\epsilon\ll 1$.

To show \eqref{eq:E'-E<}, we first show that the conditioned local density matrix $ \rho_{x_1\cdots x_{n}, \mathrm{ap}}^\alpha$ of $n+1$ after projection $e_{x_1}^{\rm ap}\otimes \cdots \otimes e_{x_1\cdots x_{n}}^{\rm ap}$ is close to the ideal one $\rho_{x_1\cdots x_{n}}^\alpha$ after projection $e_{x_1}\otimes \cdots \otimes e_{x_1\cdots x_{n}}$. From the proof of Lemma \ref{lem:clust}, each order in \eqref{eq:p=expan} can differ by \begin{equation}
    \abs{\gamma_m^{\rm ap} t^m - \gamma_m t^m } \le c_*' km \epsilon\, \lr{c_* Jt}^m,
\end{equation}
where $\gamma_m^{\rm ap}$ is the coefficient of the cluster expansion for (a matrix element of) $ \rho_{x_1\cdots x_{n}, \mathrm{ap}}^\alpha$.
Here $c_*,c_*'$ are $\order(1)$ constants, $km \epsilon$ comes from the fact that there are at most $km$ local projections involved for a cluster of size $m$, with error $\epsilon$ for each projection according to our induction hypothesis. Then the total error of the conditioned local density matrix is bounded by \begin{equation}\label{eq:proj_err}
    \sum_{m=1}^\infty c_*' km \epsilon\, \lr{c_* Jt}^m \le \epsilon\, c_*''Jt, \quad \mathrm{if}\quad c_*Jt<1,
\end{equation}
where $c_*''$ is a constant independent of $N$. Then if $Jt<1/2\max(c_*, c_*'')$, which we ensure by choosing \begin{equation}
    c_{\mathrm{M}} \le \frac{1}{2\max(c_*,c_*^{\prime\prime})},
\end{equation}
one can truncate to order $m\le \mathrm{\Theta}(\log N)$ with a truncation error $\epsilon \lr{1-c_*''Jt}>\frac{1}{2}\epsilon$, so that the \emph{computed} conditioned local density matrix $\rho_{x_1\cdots x_n}^{\alpha, \rm ap}$ (as an approximation to $\rho_{x_1\cdots x_n, \rm ap}^\alpha$) is of error $\epsilon$ (the projection error \eqref{eq:proj_err} and truncation error combined) to the ideal one $\rho_{x_1\cdots x_n}^\alpha$. Then the computed $e_{x_1\cdots x_{n+1} }^{\rm ap}$ from the conditioned local density matrix is also of error $\epsilon=\order\lr{1/\poly{N}}$ given by \eqref{eq:E'-E<}. This finishes the inductive proof.
\end{proof}

\subsection{Main result}

Proposition \ref{prop:Ex} implies the metrology protocol given in Procedure \ref{proc:metro} is efficient and achieves HL.  Therefore, we have an explicit protocol to achieve HL metrology in the presence of local interactions for arbitrary interaction graphs.

\begin{thm}\label{thm:metro}
Suppose $\ket{\GHZ}$ is evolved under Hamiltonian \eqref{eq:Hw} for time $t$.
For $t<c_{\rm M}/J$ where $c_{\rm M}$ is defined in Proposition \ref{prop:Ex}, the LOCC measurement protocol described in Procedure \ref{proc:metro} (with the basis determined by Procedure \ref{proc:Ex_clus}) achieves HL and requires only $\poly{N}$ classical computation. 
\end{thm}

\begin{proof}
We show the facts established in the 1d Theorem \ref{thm:metro1d} also apply here. 

(0) \eqref{eq:P-P=f} holds, so that the estimation of $\omega$ is unbiased at large $N$.
Following the 1d proof of \eqref{eq:P-P=f}, here we only need to show that the error of local basis \eqref{eq:E'-E<} induces $\order(1/\sqrt{N})$ error for any observable like $\mathbf{P}$.\footnote{In 1d this is established in \eqref{eq:P=P+1_N} directly by the closeness of a state to its MPS approximation \eqref{eq:phi-phi<e}. Here we do not have an approximation of the state itself.} 
We choose the right hand side of \eqref{eq:E'-E<} to scale as $N^{-3/2}$. With this bound and a similar argument of the sampling-to-computing reduction around \eqref{eq:marg-marg<}, we know the measured probability distribution $p_{\bm x}$ is close to the ideal one (with $E_{\bm{x}}$) with total variational distance $\order(1/\sqrt{N})$. This error bound also holds for observables, and is thus absorbed in the last term in \eqref{eq:P-P=f}. 

The classical computation of the following quantities are $\poly{N}$:

(1) $f(\tilde{\omega})$ for any $\tilde{\omega}\in \mathcal{I}_{\omega'}$: The proof of part (\emph{1}) of Theorem \ref{thm:metro1d} still applies, as it only relies on a Lieb-Robinson bound, which holds for general graphs \cite{review_us23}.

(2) The local basis by Procedure \ref{proc:Ex_clus} for a given trajectory $\bm x$: This is guaranteed by Proposition \ref{prop:Ex}.

(3) $\braket{\bf P}'$: This is established by \eqref{eq:overlap'>cm} and the easiness of sampling in Theorem \ref{thm:samp1}.
\end{proof}

\section{Robustness against small perturbations}\label{sec:preth}

In previous sections, we did not require weak interaction \begin{equation}\label{eq:J<<w}
    J\ll \omega.
\end{equation}
If \eqref{eq:J<<w} holds, the $Z$ field is dominant in $H_\omega$, so \eqref{eq:Z0-Z1=cin} implies the two parts $\phi^0_\omega$ and $\phi^1_\omega$ have a large energy difference with respect to $H_\omega$. By energy conservation, the two parts keep an extensive difference in $Z$ polarization \emph{forever}. Formalizing this idea, we show that the HL is robust to much longer time scales, compared to \red{$Jt< c_{\rm in}$} in previous sections (see Theorem \ref{thm:phase}).

\begin{thm}\label{thm:weak_perturb}
Suppose \begin{equation}
    2J < c_{\rm in}\omega.
\end{equation}
Then if \eqref{eq:Z0-Z1=cin} and \eqref{eq:initial_cor} are satisfied by the initial state, \eqref{eq:partial_phi} and \eqref{eq:psi'_omega} hold for any constant time $t=\order(N^0)$ that does not scale with $N$, with \begin{equation}\label{eq:partial_phi_weak}
    c_\omega\ge 
    t \lr{ c_{\rm in} - 2J/\omega }, \quad \mathrm{and} \quad 2\lr{ c_{\rm in} - 2J/\omega } tN\le \partial_\omega f(\omega)\le 2tN.
\end{equation}
\end{thm}

\begin{proof}
By energy conservation $\ee^{\ii s H_\omega} (\omega Z + V) \ee^{-\ii s H_\omega} = \omega Z + V$, \begin{equation}\label{eq:Z-Z<2NJ}
    \norm{\ee^{\ii s H_\omega} Z \ee^{-\ii s H_\omega} - Z} = \norm{\frac{1}{\omega}\lr{V-\ee^{\ii s H_\omega} V \ee^{-\ii s H_\omega}} } \le \frac{2}{\omega} \norm{V} \le 2NJ/\omega.
\end{equation}

Combining \eqref{eq:Z-Z<2NJ} with \eqref{eq:>N-Z-Z} and the first line of \eqref{eq:Z-Z<} yields \begin{equation}\label{eq:phase_weak}
    \ii \lr{ \bra{\phi^0_\omega}\partial_\omega \ket{\phi^0_\omega}-\bra{\phi^1_\omega}\partial_\omega \ket{\phi^1_\omega}} \ge 2\lr{c_{\rm in}- 2J/\omega }tN.
\end{equation}
Substituting \eqref{eq:0-1>N} by \eqref{eq:phase_weak}, the rest of the proof of Theorem \ref{thm:weak_perturb} follows verbatim to establish Theorem \ref{thm:weak_perturb}. In particular, the errors $\order(\sqrt{N})$ in \eqref{eq:partial_phi} and $\order(1/\sqrt{N})$ in \eqref{eq:psi'_omega} come from the fact that $g(Jt)$ in \eqref{eq:cor<g} does not scale with $N$, which holds for any constant $t=\order(N^0)$. 
\end{proof}

Theorem \ref{thm:weak_perturb} implies HL for QFI following Corollary \ref{cor:QFI}. In terms of achieving HL by the LOCC measurement protocol, 1d results in Section \ref{sec:1d} continue to hold here for any constant $t$ independent of $N$, because the two parts of the state remain to be MPS with $\poly{N}$ bond dimensions. However, efficiency of the protocol in higher dimensions is not guaranteed if $Jt\gtrsim 1$.

In the proof of Theorem \ref{thm:weak_perturb}, we used bound \eqref{eq:Z-Z<2NJ} that simply comes from energy conservation. If the system is finite-dimensional and \eqref{eq:J<<w} holds, there is actually another mechanism called prethermalization \cite{preth_rigor17,preth_exp_Wei,preth_exp_Peng,preth_long20} that yields a similar bound \eqref{eq:Z-Z<preth} below. In a nutshell, although the $\mathrm{U}(1)$ rotational symmetry along the $z$ direction is explicitly broken by the interaction $V$, prethermalization theory proves that there is an approximately conserved $\mathrm{U}(1)$-charge (which is a dressed version of $Z$) on time scales  $t \ll t_{\rm pre}$, where the prethermalization time (after which the system thermalizes with no conserved $\mathrm{U}(1)$-charge) \begin{equation}\label{eq:tpre}
    t_{\rm pre}:= J^{-1}\exp(c_{\rm pre}\, \omega/J),
\end{equation}
which is exponentially large in $\omega/J\gg 1$. 
Here $c_{\rm pre}$ is a constant independent of $N,\omega,J$.

Physically, prethermalization roughly says that an initial eigenstate of $Z$ cannot decay into other charge sectors in short time. This is because the other charge sectors have at least $\Delta E =\omega$ energy difference in time-dependent perturbation theory of $J/\omega$, so at any finite order the decay channel is off-resonant and ineffective. To overcome this energy gap by gaining energy from the perturbation $V$, one has to go to non-perturbatively large order $k_*\sim \omega/J$, where the local decay rate is exponentially small in $k_*$. This leads to the large prethermalization time \eqref{eq:tpre} for $H_\omega$. For more general class of models with a many-body gap, this dynamical stability under perturbation is made precise in \cite{our_preth22}, albeit with a slightly weaker bound on $t_{\rm pre}$.

Applying prethermalization to the metrology setting, we have
\begin{equation}\label{eq:Z-Z<preth}
    \norm{\ee^{\ii s H_\omega} Z \ee^{-\ii s H_\omega}-Z}\le c_{\rm pre}^Z NJ/\omega, \quad \forall s \le t_{\rm pre},
\end{equation}
from combining Theorem 3.1 and 3.2 in \cite{preth_rigor17}, where
the constant $c_{\rm pre}^Z$ is determined by geometry of the interaction graph. If $c_{\rm pre}^Z\le 2$, \eqref{eq:Z-Z<preth} would be a stronger bound than \eqref{eq:Z-Z<2NJ} in the prethermalization time window, which further gives a tighter bound than \eqref{eq:partial_phi_weak} on the slope of $f(\omega)$.

In higher dimensions, it is interesting to ask whether an evolved product state is easy to sample before prethermalization time $t_{\rm pre}$. This might be tackled by combining the Schrieffer-Wolff transformation involved in prethermalization \cite{our_preth22} and the cluster expansion techniques in Section \ref{sec:samp}. If such a polynomial-time classical sampler exists, we expect the LOCC measurement protocol would be efficient to implement in higher dimensions, for $t<t_{\rm pre}$.

\section{Outlook}

In this paper, we proved the robustness of the Heisenberg limit: starting with a GHZ state, unwanted (but known) interactions do not spoil our ability to perform Heisenberg-limited metrology.  We moreover presented a provably efficient LOCC protocol that requires only polynomial classical computational resources to achieve the HL. 

It will be interesting to see whether our results generalize to (\emph{i}) other paradigms of quantum metrology, e.g., using spin-squeezed states \cite{squeez_93,squeez_rev11} that do not reach the HL generated by 
 microscopic Hamiltonians; (\emph{ii}) long-range interaction $V$ \cite{squeez_Ising16,squeez_Rey20,squeez_Norm23}; (\emph{iii}) multi-parameter quantum metrology \cite{metro_multip_rev16}. Since GHZ states have been created in experiments up to $\sim 20$ qubits \cite{GHZ_photon18,GHZ_Wei19,GHZ_SC19,GHZ_Rydberg19,GHZ_ion21}, it is possible to demonstrate our protocol in near-term quantum devices. 

Part of our result relied on a polynomial-time classical algorithm to sample from a short-time evolved state, assuming the measurement basis is far from the computational basis. We leave as future work whether one can get rid of this assumption on the sampling basis. Our sampling algorithm may also be useful as a subroutine for other applications in classical/quantum computation: for example, it has been proposed to use real-time evolution to study quantum thermal states at finite temperature \cite{hybrid_T21,hybrid_T23,hybrid_T23_1}.

\section*{Acknowledgements}
We thank Alexander Schuckert, Sisi Zhou, Jacob Bringewatt and Daniel J. Spencer for valuable comments.
This work was supported by the Alfred P. Sloan Foundation under Grant FG-2020-13795 (AL) and by the U.S. Air Force Office of Scientific Research under Grant FA9550-21-1-0195 (CY, AL).

\begin{appendix}

\section{Failure of the ideal protocol under perturbation}\label{app:A}

\begin{prop}\label{prop:X_fail}
    Suppose $V=JX=J\sum_j X_j$, and the state \eqref{eq:psiw} after evolution is projected onto the $x$-basis. Let $p_{\bm{x}}$ be the probability to get measurement outcome $\bm{x}$. Define another probability distribution $\tilde{p}_{\bm{x}}$ from performing the same metrology protocol but on a mixed initial state \begin{equation}\label{eq:trho=}
        \tilde{\rho}_{\rm in}:= \frac{1}{2} \lr{\ket{0\cdots 0}\bra{0\cdots 0} + \ket{1\cdots 1}\bra{1\cdots 1}}.
    \end{equation} 
    Then, for generic $t\neq n\pi/\sqrt{\omega^2+J^2}$ with integer $n$, $p_{\bm{x}}$ is exponentially close to $\tilde{p}_{\bm{x}}$ in total variation distance \begin{equation}\label{eq:P-tP<}
        \norm{p-\tilde{p}}_1:= \sum_{\bm{x}} \abs{p_{\bm{x}} - \tilde{p}_{\bm{x}} } = \exp\mlr{-\Omega(N)}.
    \end{equation}
\end{prop}

We will prove this Proposition shortly.
Assuming the number of repeated measurements $M$ is at most polynomial in $N$ (i.e., number of samples of $\bm{x}$ is polynomial),
\eqref{eq:P-tP<} implies that one cannot distinguish the two distributions. The true experiment is done as if the initial state is replaced by the mixed state \eqref{eq:trho=}, where the two parts evolve incoherently and the extensive phase difference is lost. Thus the $x$-basis measurement cannot perform HL and returns to SQL, at least for generic $t$. At the specific integer values of $t$, the state returns to the inital GHZ state; for more general $V$ that couples spins, we expect SQL for any $t>0$ due to the absence of such many-body recurrence.

\begin{proof}[Proof of Proposition \ref{prop:X_fail}]
For the chosen $V$, $H_\omega = \omega Z+JX$ acts on the spins individually. Each spin just rotates along a tilted magnetic field with angular frequency $2\sqrt{\omega^2+J^2}$, so \begin{equation}\label{eq:varphi}
    \ket{\psi_\omega} = \frac{1}{\sqrt{2}} \lr{\ket{\varphi_0}^{\otimes N} + \ket{\varphi_1}^{\otimes N} },
\end{equation} 
where $\{\ket{\varphi_0}, \ket{\varphi_1}\}$ is an instantaneous local basis for a single spin. Note that the relative phase difference between the two terms in \eqref{eq:varphi} is absorbed in the local basis. Then \begin{align}\label{eq:P=tP+}
    p_{\bm{x}} &= \abs{\braket{\bm{x}|\psi_\omega}}^2 = \frac{1}{2} \mlr{ \braket{ \varphi_0^{\otimes N}| \bm{x} } \braket{\bm{x}| \varphi_0^{\otimes N} } + \braket{ \varphi_1^{\otimes N}| \bm{x} } \braket{\bm{x}| \varphi_1^{\otimes N} } + \lr{ \braket{ \varphi_0^{\otimes N}| \bm{x} } \braket{\bm{x}| \varphi_1^{\otimes N} } + \mathrm{c.c.} } } \nonumber\\
    &= \tilde{p}_{\bm{x}} + \frac{1}{2} \lr{ \braket{ \varphi_0^{\otimes N}| \bm{x} } \braket{\bm{x}| \varphi_1^{\otimes N} } + \mathrm{c.c.} },
\end{align}
because $\tilde{p}_{\bm{x}}$ is from measuring state $\lr{\ket{\varphi_0}\bra{\varphi_0}^{\otimes N} + \ket{\varphi_1}\bra{\varphi_1}^{\otimes N} }/2$.
Bounding the second term in \eqref{eq:P=tP+} involves \begin{equation}
    \abs{\braket{\bm{x}|\varphi_\alpha^{\otimes N} } } = \abs{ \braket{+|\varphi_\alpha}^{N_+} \braket{-|\varphi_\alpha}^{N_-} } = \eta_\alpha^{N_+} \lr{\sqrt{1-\eta_\alpha^2}}^{N_-},
\end{equation}
where $\{\ket{+},\ket{-}\}$ is the local $x$-basis, and $N_+$ ($N-=N-N_+$) is the number of $+$ ($-$) in string $\bm{x}$. We have defined $\eta_\alpha=\abs{\braket{+|\varphi_\alpha}}$ so that $\abs{\braket{-|\varphi_\alpha}}=\sqrt{1-\abs{\braket{+|\varphi_\alpha}}^2}=\sqrt{1-\eta_\alpha^2}$. Similarly the two $\eta$s are related by $\eta_1=\sqrt{1-\eta_0^2}$. As a result, \eqref{eq:P=tP+} implies \begin{equation}
    \sum_{\bm{x}} \abs{p_{\bm{x}} - \tilde{p}_{\bm{x}} } \le \sum_{\bm{x}} \eta_0^{N_+} \lr{\sqrt{1-\eta_0^2}}^{N_-} \eta_1^{N_+} \lr{\sqrt{1-\eta_1^2}}^{N_-} = \sum_{\bm{x}} \lr{\eta_0\sqrt{1-\eta_0^2}}^N = \lr{\mathsf{f}(\eta_0)}^N,
\end{equation}
where $\mathsf{f}(\eta_0):=2\eta_0\sqrt{1-\eta_0^2}\le 1$, with equality only at $\eta_0=1/\sqrt{2}$. For generic $t\neq n\pi/\sqrt{\omega^2+J^2}$, an initial $\ket{0}$ is rotated to point towards $\ket{+}$ more (comparing to $\ket{-}$), so $\eta_0=\abs{\braket{+|\varphi_0}}>1/\sqrt{2}$, which leads to \eqref{eq:P-tP<}.
\end{proof}

\section{Measurement protocol based on undoing the perturbation} \label{app:B}

In the main text,
Section \ref{sec:QFI} establishes robustness of HL at the level of QFI, while the later sections develop an LOCC measurement protocol achieving HL that, under certain conditions, is efficient to implement. In this Appendix, we provide another such protocol, which requires quantum controls beyond single-qubit that effectively reverses the evolution caused by interaction $V$. For simplicity we assume the initial state is the ideal GHZ state.

Recall that without $V$, the measurement is simply done for observable $\fX$, the product of $X_j$s in \eqref{eq:Xid}. We try to deform the observable according to $V$. \begin{align}\label{eq:cos=Xw}
    \cos(2N\omega t) &=\alr{\fX}_{\id} = \bra{\GHZ} \ee^{\ii t \omega Z}\fX \ee^{-\ii t \omega Z} \ket{\GHZ} = \bra{\GHZ}\ee^{\ii t H_\omega}\ee^{-\ii t H_\omega} \ee^{\ii t \omega Z}\fX \ee^{-\ii t \omega Z}\ee^{\ii t H_\omega}\ee^{-\ii t H_\omega} \ket{\GHZ} \nonumber\\ &= \bra{\psi_\omega} \fX_\omega \ket{\psi_\omega},
\end{align}
where \begin{equation}\label{eq:Xw}
    \fX_\omega=R_\omega^\dagger\fX R_\omega = \prod_j X_{j,\omega}, \quad R_\omega := \ee^{-\ii t \omega Z}\ee^{\ii t H_\omega}
\end{equation}
is a product of almost local operators, because $X_{j,\omega}=R_\omega^\dagger X_j R_\omega$ evolved by unitary $R_\omega$ is mostly supported in a finite set of spins near site $j$. If one measures observable $\fX_\omega$ in the final state $\ket{\psi_\omega}$, then it is equivalent to measure $\fX$ in the unperturbed protocol. In practice, one can apply unitary $R_\omega$ to $\ket{\psi_\omega}$ to effectively time-reverse the dynamics, and then measure $\fX$ by projective measurement to the $x$-basis.

However, the above measurement does only local estimation because $\fX_\omega$ depends on the unknown parameter $\omega$. For global estimation, instead, we propose to measure observable $\fX_{\omega'}$ determined by prior knowledge $\omega'$ \eqref{eq:w'=}. By bounding the difference between $\fX_{\omega'}$ and $\fX_\omega$, we show such measurement achieves HL for \red{sufficiently small $Jt$}.
\begin{prop}
The measurement observable $\bra{\psi_\omega} \fX_{\omega'} \ket{\psi_\omega}$ satisfies
    \begin{equation}\label{eq:X'-X}
        \abs{\bra{\psi_\omega} \fX_{\omega'} \ket{\psi_\omega} - \cos(2N\omega t)}\le \norm{\fX_{\omega'} - \fX_\omega } \le \pi Jt/2.
    \end{equation}
    Moreover, \begin{equation}\label{eq:partial_Xw'}
        \abs{\partial_\omega \bra{\psi_\omega} \fX_{\omega'} \ket{\psi_\omega} } \ge Nt \mlr{ 2\abs{\sin(2N\omega t)} - (\pi+2) J t} ,
    \end{equation}
    which leads to HL for \begin{equation}\label{eq:t<sinJ}
        t<2\abs{\sin(2N\omega t)}/(\pi+2) J,
    \end{equation}
    because the denominator in \eqref{eq:HL} is \red{$\abs{\partial_\omega \bra{\psi_\omega} \fX_{\omega'} \ket{\psi_\omega} }^2 =\Theta( N^2) $}. 
\end{prop}

\begin{proof}
From \eqref{eq:Xw} and $\norm{\fX}=1$, we have \begin{align}\label{eq:X'-X<}
    \norm{\fX_{\omega'} - \fX_\omega } &\le 2\norm{\ee^{-\ii t \omega' Z}\ee^{\ii t H_{\omega'}} - \ee^{-\ii t \omega Z}\ee^{\ii t H_{\omega}} } = 2\norm{ \mathcal{T} \ee^{\ii \int^t_0 \dd s V'(s)} - \mathcal{T} \ee^{\ii \int^t_0 \dd s V(s)} }\nonumber\\
    &= 2\norm{ \mathcal{T}\int^t_0 \dd s\, \ee^{\ii \int^t_s \dd s' V'(s')} [V'(s)-V(s)]\ee^{\ii \int^s_0 \dd s' V(s')} } \le 2\int^t_0\dd s\, \norm{V'(s)-V(s)} \nonumber\\
    &\le 2\int^t_0\dd s\, \int^s_0\dd s'\, \norm{[(\omega'-\omega)Z, V(s')]} = 2\int^t_0\dd s\, \int^s_0\dd s'\, \abs{\omega'-\omega} \norm{[Z,V]} \nonumber\\
    &\le 2t^2 \abs{\omega'-\omega} NJ.
\end{align}
In the first line we have used the interaction picture \begin{equation}
    \ee^{-\ii t \omega Z}\ee^{\ii t (V+\omega Z)} = \mathcal{T} \ee^{\ii \int^t_0 \dd s V(s)}, \where V(s) = \ee^{-\ii s \omega Z}V\ee^{\ii s \omega Z},
\end{equation}
and $V'(s)$ is defined similarly with $\omega$ replaced by $\omega'$. In the second and third lines of \eqref{eq:X'-X<}, we have used the Duhamel identity. \eqref{eq:X'-X} then follows from \eqref{eq:w-w'}.

To get \eqref{eq:partial_Xw'}, \begin{align}
    \abs{\partial_\omega \bra{\psi_\omega} \fX_{\omega'} \ket{\psi_\omega}}
    &= \abs{\partial_\omega \bra{\psi_\omega} \fX_{\omega} \ket{\psi_\omega} + \partial_\omega\bra{\psi_\omega} \fX_{\omega'}-\fX_{\omega} \ket{\psi_\omega} } \nonumber\\
    &\ge \abs{\partial_\omega \bra{\psi_\omega} \fX_{\omega} \ket{\psi_\omega} } - \norm{\partial_\omega \lr{ \fX_{\omega'}-\fX_{\omega} } } -2 \norm{\fX_{\omega'} - \fX_\omega} \norm{\ket{\partial_\omega \psi_\omega}} \nonumber\\
    &\ge 2Nt\abs{\sin(2N\omega t)} - 2t^2NJ- \pi Jt \cdot tN = Nt \mlr{ 2\abs{\sin(2N\omega t)} - (\pi+2) J t}.
\end{align}
Here in the last line, we have used the Lipschitz property \eqref{eq:X'-X<} that holds for any $\omega,\omega'$ to bound $\norm{\partial_\omega \lr{ \fX_{\omega'}-\fX_{\omega} } }= \norm{\partial_\omega \fX_{\omega} } \le 2t^2NJ$, and \eqref{eq:X'-X} together with $\norm{\ket{\partial_\omega \psi_\omega}}\le tN$ from \eqref{eq:cZ=Z} for the last term.
\end{proof}

The time window \eqref{eq:t<sinJ} is $\order(J^{-1})$ for almost all $\omega\in \mathcal{I}_{\omega'}$ except if $\omega$ is very close to $\omega'\pm \pi/(4Nt)$; the latter case can be avoided by slightly adjusting $N$ or $t$.
After measuring $g(\omega):=\bra{\psi_\omega} \fX_{\omega'} \ket{\psi_\omega} $ quantumly, one cannot compare it to the ideal function $\cos(2N\omega t)$ to read out $\omega$ due to the error \eqref{eq:X'-X}. Instead, one needs to classically calculate $g(\omega)$ as a function of $\omega$ to compare the quantum result with.
In general dimensions, there are algorithms with runtime being quasi-polynomial in $N$ ($\poly{N}$ for $\le 2$d) \cite{compu_prob19,compu_prob21,compu_prob22,comp_prob_H23} for such classical simulation, known as the quantum mean value problem, if \red{$Jt=\order(\log N$)}. To conclude, the $\mathbf{X}_{\omega'}$ measurement protocol satisfies HL and requires (quasi-)polynomial classical computation to implement.

\end{appendix}

\bibliographystyle{quantum}
\bibliography{biblio}

\end{document}